\documentclass[letterpaper,11pt]{article}
\usepackage{amsthm}  
\usepackage{amsmath, amssymb}
\usepackage{color}
\usepackage{colortbl}
\usepackage{fullpage}
\usepackage[numbers]{natbib}
\usepackage{algorithm}
\usepackage[noend]{algpseudocode}
\usepackage{thmtools, thm-restate}
\usepackage{tikz}
\usepackage{enumerate}
\usepackage{thm-restate}
\usepackage{graphicx}
\usepackage{amsopn}
\usepackage{caption}
\usepackage{hyperref}
\usepackage{subcaption}
\usepackage{zerosum,csquotes}
\usepackage{enumitem,linegoal}
\usepackage{comment}
\usepackage{float}
 \usepackage[super]{nth}
\usepackage{subcaption}
\usepackage{authblk}

\usepackage{ctable}					

\usepackage{mathtools}
\usepackage[flushleft]{threeparttable}
\usepackage[normalem]{ulem}
\usepackage{footnote}

\makesavenoteenv{tabular}
\makesavenoteenv{table}

\usepackage[margin=1in]{geometry}

\newtheorem{theorem}{Theorem}[section]
\newtheorem{lemma}[theorem]{Lemma}
\newtheorem{claim}[theorem]{Claim}

 \newtheorem{question}{Question}
\theoremstyle{definition}
 \newtheorem{definition}[theorem]{Definition}
\newtheorem{example}[theorem]{Example}
\newtheorem{observation}[theorem]{Observation}

\makeatletter
\newif\ifqed
\def\GrabProofArgument[#1]{ #1: \egroup\ignorespaces}
\def\proof{\noindent\textbf\bgroup Proof%
	\@ifnextchar[{\GrabProofArgument}{. \egroup\ignorespaces}\global\qedtrue}

\def\qedhere{\ifmmode\tag*{\qedsign}\else\hspace*{\fill}\qedsign\medskip\fi\global\qedfalse}
\def\qedsign{$\Box$}
\makeatother

\newcommand{\agents}{\mathcal{N}}
\newcommand{\items}{\mathcal{M}}
\newcommand{\ite}{b}
\newcommand{\allocation}{\mathcal{A}}
\newcommand{\ralloc}{\mathcal{R}}
\newcommand{\falloc}{\mathcal{F}}

\newcommand{\domain}{D}
\newcommand{\dom}{D}

\newcommand{\valu}{v}

\newcommand{\argmax}{\text{argmax}}

\newcommand{\MMS}{\mathsf{MMS}}
\newcommand{\XOS}{\mathsf{XOS}}

\newcommand{\etal}{\emph{et al.}}

\definecolor{mygreen}{RGB}{20,140,80}
\definecolor{mylightgray}{RGB}{230,230,230}

\definecolor{mygreen}{RGB}{20,140,80}
\definecolor{mydarkgray}{gray}{0.15} 
\definecolor{oceanblue}{HTML}{2c55c2}
\hypersetup{
     colorlinks=true,
     citecolor= mygreen,
     linkcolor= black
}

\newenvironment{claimproof}[1]{\emph{Proof. }\space#1}{\hfill $\blacksquare$\medskip\par}

\newcounter{proccnt}

\newcommand{\konote}[1]{}

\title{Randomized and Deterministic Maximin-share Approximations for Fractionally Subadditive Valuations}

\author[1,3]{Hannaneh Akrami}
\author[1,4]{Kurt Mehlhorn}
\author[2]{Masoud Seddighin}
\author[1,3]{Golnoosh Shahkarami}

\affil[1]{Max Planck Institute for Informatics}
\affil[2]{Tehran Institute for Advanced Studies}
\affil[3]{Graduiertenschule Informatik, Universit\"at des Saarlandes}
\affil[4]{Universit\"at des Saarlandes}

\begin{document}
	\newcommand{\ignore}[1]{}
\renewcommand{\theenumi}{(\roman{enumi})}
\renewcommand{\labelenumi}{\theenumi.}
\sloppy

%
%

\date{}

\maketitle


\begin{abstract}
We consider the problem of guaranteeing maximin-share ($\MMS$) when allocating a set of indivisible items to a set of agents with fractionally subadditive ($\XOS$) valuations.  
For $\XOS$ valuations, it has been previously shown that for some instances no allocation can guarantee a fraction better than $1/2$ of maximin-share to all the agents. Also, a deterministic allocation exists that guarantees $0.219225$ of the maximin-share of each agent. 
Our results involve both deterministic and randomized allocations. On the deterministic side, we improve the best approximation guarantee for fractionally subadditive valuations to $3/13 = 0.230769$. We develop new ideas on allocating large items in our allocation algorithm which might be of independent interest. Furthermore, we investigate randomized algorithms and the Best-of-both-worlds fairness guarantees. We propose a randomized allocation that is $1/4$-$\MMS$ ex-ante and $1/8$-$\MMS$ ex-post for $\XOS$ valuations. Moreover, we prove an upper bound of $3/4$ on the ex-ante guarantee for this class of valuations.  
\end{abstract}

\section{Introduction}

Fair allocation is a central problem in economics since decades. It arises naturally in real-world applications such as advertising, negotiation, rent sharing, inheritance, etc \cite{caragiannis2019unreasonable,dehghani2018envy,dickerson2014computational,Foley:first,nash1950bargaining,varian1973equity}. In discrete fair division, the basic scenario is that we want to distribute a set $\items$ of $m$ indivisible items among $n$ agents, such that the allocation is deemed fair by the agents. Each agent $i$ has a valuation function $\valu_i: 2^\items \rightarrow \mathbb{R}^+$ that represents her happiness for receiving a subset of items.

How do we evaluate fairness? This question has been the subject of intense debates in various contexts, including philosophy, economics, distributive justice, and mathematics. Since the introduction of the cake-cutting\footnote{Cake-cutting 
is the continuous version of the fair allocation problem where the resource is a single divisible cake.} problem by Hugo Steinhaus in 1942 \cite{Steinhaus:first}, scientists have suggested several different notions to evaluate fairness. The challenge is that a fairness notion must be both reasonable in terms of justice and implementable in practice. Conceived by Steinhaus, proportionality is one of the most natural and prominent notions. An allocation is \emph{proportional} if the share allocated to every agent is worth at least $1/n$ of her value for the entire resource.

Unfortunately,  despite many positive results on proportionality for cake-cutting, simple examples indicate that proportionality is not a proper fairness criterion for the case of indivisible items. For example, when there is one item and two agents, one agent receives nothing, though her proportional share is non-zero assuming that the item has positive value for both of the agents. An alternative form of proportionality adopted to deal with indivisibilities is \emph{Maximin-share} ($\MMS$) introduced by Budish~\cite{Budish:first}. For every agent $i$, the maximin-share of agent $i$, denoted by $\MMS_i$ is defined as follows:
$$
\MMS_i = \max_{\langle \Pi_1,\Pi_2,\ldots, \Pi_n\rangle \in \Omega}\, \min_{1 \leq j \leq n} \valu_i(\Pi_j),
$$
where $\Omega$ is the set of all partitions of $\items$ into $n$ parts. 
It is known that even when all the valuations are additive, 
$\MMS$ allocations, in which each agent $i$ gets at least $\MMS_i$, need not exist~\cite{ feige2022tight, kurokawa2018fair, 10.1145/2600057.2602835}.
Nevertheless, several studies in recent years show that it is possible to guarantee a constant factor of her maximin-share to each agent for various classes of valuation functions, including additive, submodular, and fractionally subadditive.\footnote{We refer to Section \ref{sec:prelim} for a formal definition of these valuation classes.} See Table \ref{sum} for the state-of-the-art guarantees on the maximin-share.

\begin{table}[h]
		\centering
		\sffamily
		\scriptsize
		\begin{tabular}{|c|c|c|}
			\hline \rowcolor{lightgray}
			\rule[-0.0ex]{0pt}{2.5ex}  Valuation Class & Approximation Guarantee & Upper bound \\ 
			\hline \hline 
			\rule[0.ex]{0pt}{2.5ex} Additive &  $\frac{3}{4} + \frac{3}{3836}$  \cite{Akrami2023BreakingT3}& $\frac{39}{40}$ \cite{feige2022tight} \\ 
			\hline 
			\rule[0.0ex]{0pt}{2.5ex} Submodular & $\frac{10}{27}$ \cite{uziahu2023fair} & $\frac{3}{4}$ \cite{ghodsi2018fair}\\ 
			\hline 
			\rule[0.0ex]{0pt}{2.5ex} Fractionally Subadditive & $0.219225$  \cite{seddighin2022improved}& $\frac{1}{2}$ \cite{ghodsi2018fair}\\ 
			\hline 
			\rule[0.ex]{0pt}{2.5ex}Subadditive & $\frac{1}{\log n \log \log n}$ \cite{seddighin2022improved} & $\frac{1}{2}$ \cite{ghodsi2018fair} \\ 
			\hline 
		\end{tabular} 
		\vspace{0.1in}
		
	\caption{A summary of the results for $\MMS$ in different valuation classes.}
	\label{sum}
\end{table}

Let us revisit the instance with one item and two agents. Suppose that the item has value $6$ for both agents. By definition, the proportional share of each agent is $6/2 = 3$, and since one agent receives no item, satisfying proportionality or any approximation of it is impossible. On the other hand,  we have $\MMS_1 = \MMS_2 = 0$. Thus, allocating the item to any agent satisfies maximin-share. Indeed, we can circumvent the non-guaranteed existence of a fair allocation by reducing our expectation of fairness to the maximin-share. However, regardless of how we allocate the item, one agent receives one item, and the other receives nothing. Therefore, having one agent with zero utility is inevitable for any deterministic allocation in this example. The question then arises: Can we do better? 

One way to improve the allocation is to use randomization and obtain a better guarantee in expectation (\emph{ex-ante}). 
For example, we can allocate the item to each agent with probability $1/2$. This way, the expected utility of each agent is equal to $3$. In economic terms, this allocation satisfies proportionality \emph{ex-ante}. Note that one agent receives no item \emph{ex-post} (that is, after fixing the outcome); however, it guarantees proportionality ex-ante to both agents. 

Considering random allocations and ex-ante fairness makes the problem much handier. For instance, assuming there are $n$ items and $n$ agents, allocating each item to each agent with probability $1/n$ satisfies proportionality ex-ante. However, this randomized allocation has no ex-post fairness guarantee: with a non-zero probability, the outcome allocates all the items to one agent, and the rest of the agents receive no item. It is desirable to find allocations with simultaneous ex-ante and ex-post guarantees. The support of such an allocation is limited to outcomes with some desirable fairness guarantee. For example, consider the following random allocation: we choose a random permutation of the items and allocate the $i^{\textrm{th}}$ item in the permutation to agent $i$. This allocation satisfies proportionality ex-ante and maximin-share ex-post. 

This approach is first introduced by Aziz~\cite{DBLP:journals/corr/abs-2002-10171}. Recently, several studies have investigated randomized allocations with both ex-ante and ex-post guarantees. Some notable results with the focus on additive valuations are (i) an ex-ante envy-free and ex-post EF1 allocation algorithm \cite{freeman2020best}, (ii) an ex-ante proportional and ex-post $1/2$-$\MMS$ allocation algorithm \cite{babaioff2022best}, and (iii) an ex-post $3/4$-MMS and ex-ante $0.785$-MMS allocation algorithm \cite{Akrami2023ImprovingAG}. Very recently, Feldman \etal~\cite{cycle-breaking} studied best-of-both-worlds for subadditive valuations and gave an allocation algorithm with ex-ante guarantee of $1/2$-envy-freeness and ex-post guarantee of $1/2$-EFX and EF1. 

In this paper, we explore fair deterministic and randomized allocations for fractionally subadditive valuation functions. A valuation function $\valu_i(\cdot)$ is fractionally subadditive ($\XOS$),  if there exists a family of additive valuation functions $u_{i,1},u_{i,2}, \ldots, u_{i,\ell} : 2^\items \rightarrow \mathbb{R}_{\geq 0}$ such that for every set $S$ we have 
	$$
	\valu_i(S) = \max_{1 \leq k \leq \ell} u_{i,k}(S).
	$$ 
Fractionally subadditive is more general than additive, gross substitute, and submodular; it is less general than subadditive set functions.

Our fairness notion is maximin-share. We are looking for allocations that satisfy an approximation of maximin-share both ex-ante and ex-post for fractionally subadditive valuations. Of course, the ex-ante guarantee will be at least as strong as the ex-post one. We will also improve the ex-post guarantee (deterministically) for this class of valuations.
We refer to the next subsection for an overview on our results and techniques.  

\subsection{Our Results and Techniques}\label{sec:results}
In this paper, we provide improved approximation guarantees for the maximin-share in the fractionally subadditive setting. We investigate randomized and deterministic allocation algorithms. MMS is scale-invariant and in the rest of this paper, without loss of generality, we assume $\MMS_i=1$ for all agents $i$.

\subsection*{Randomized Allocations} For randomized allocations, we take one step toward extending the best-of-both-worlds idea for fairness concepts to valuations more general than additive. For the additive setting, Babaioff \etal~\cite{babaioff2022best} proved the existence of randomized allocations that are proportional ex-ante and $1/2$-$\MMS$ ex-post. However, as we show in Section \ref{sec:exante}, though guaranteeing proportionality ex-ante is easy for valuations such as submodular, $\XOS$, and subadditive, this notion is not always a proper choice as a fairness criterion. In fact, for some instances, proportionality can be as small as $O(1/n)$ of the $\MMS$ value, which is highly undesirable. Therefore, here we focus on guaranteeing $\MMS$ approximations both ex-ante and ex-post. More precisely, we are looking for randomized allocations that are $\alpha$-$\MMS$ ex-ante and $\beta$-$\MMS$ ex-post, where $0<\beta< \alpha $. Very recently, Akrami et al. \cite{Akrami2023ImprovingAG} studied the same question when all agents have additive valuations and proved the existence of a randomized allocation which is $3/4$-MMS ex-post and $0.785$-MMS ex-ante.

In contrast to the additive setting, for $\XOS$ valuations guaranteeing $\MMS$ ex-ante is not easy. Recall that in the additive setting, a fractional allocation that allocates a fraction $1/n$ of each item to each agent is proportional and consequently $\MMS$. However, for some $\XOS$ instances this allocation is $O(1/n)$-$\MMS$ (Observation \ref{obs:2}). Indeed, we show that there are instances for which guaranteeing $\MMS$ ex-ante is not possible. More precisely, we show that there are instances such that no randomized allocation can guarantee a factor better than $3/4$ of her maximin-share to each agent (Lemma \ref{lem:1}).

On the positive side, we propose an algorithm that finds a randomized allocation which is $1/4$-$\MMS$ ex-ante.
Furthermore, by leveraging additional innovative ideas, we extend this result to encompass both ex-ante and ex-post guarantees.

\begin{restatable}{theorem}{exante}\label{thm:1}
    For any instance with $\XOS$ valuations, there exists a randomized allocation that is $1/4$-$\MMS$ ex-ante. 
\end{restatable}
	
The idea to prove the approximation guarantee of our allocation is inspired by the work of Ghodsi \etal~\cite{ghodsi2018fair}.
In fact, we show that the fractional allocation $\falloc$ which is obtained from  the following program is $1/4$-$\MMS$:
\begin{align}
	\text{maximize }   &\sum_{1 \leq i \leq n}    u_i  \nonumber      \\
	\text{subject to } &  \sum_{1 \leq i \leq n} f_{i,j} = 1& \forall_j  \nonumber\\
	&  f_{i,j} \geq 0& \forall_{i,j} \nonumber \\
	&u_i = \min(\frac{1}{2}, \max_k \sum_j u_{i,k}(\{\ite_j\}) f_{i,j}). 
	& \forall_{i}  \label{opt:1}
\end{align}
Intuitively, Program \ref{opt:1} defines an alternative valuation function $\bar{\valu}_i(\cdot)$ for each agent $i$ and then finds an allocation that maximizes social welfare with respect to these valuations. For every $i$, $\bar\valu_i(\cdot)$ is the same as $\valu_i(\cdot)$, except that for every bundle $X$ with $\valu_i(X)>1/2$ we have $\bar \valu_i(X)=1/2$. From an economical standpoint, one can see the answer of this program as an interesting trade-off between fairness and social welfare.

We can then convert the fractional allocation $\falloc$ into a randomized one (Theorem \ref{thm:fei} and Lemma \ref{randomized versus fractional}). 
However, there is no non-trivial guarantee on the fairness of the ex-post allocation. To resolve this issue, we add one additional step to our algorithm (namely allocating single large items), and also add another constraint to the optimization program. Before solving the optimization problem, we check to see if a single item can satisfy an agent. The main goal of this step is to make sure that the value of each remaining item for the remaining agents is small enough so that we can use Theorem \ref{thm:fei} to convert the fractional allocation into a randomized one with an ex-post guarantee. 

We also add another constraint to the optimization problem to obtain the following half-integral optimization program:

\begin{align}
	\text{maximize }   &\sum_{1 \leq i \leq n}    u_i  \nonumber      \\
	\text{subject to } &  \sum_{1 \leq i \leq n} f_{i,j} = 1& \forall_j  \nonumber\\
	&  f_{i,j} \in \{0,1/2,1\} & \forall_{i,j} \nonumber \\
	&u_i = \min(\frac{1}{2}, \max_k \sum_j u_{i,k}(\{\ite_j\}) f_{i,j}).
	& \forall_{i}  \label{opt:3}	
\end{align}
Despite this additional step and constraint, we prove that the answer of Program \ref{opt:3} also gives the $1/4$-$\MMS$ ex-ante approximation guarantee.  
However, note that the upper bound on the value of items obtained from the additional step combined with Theorem \ref{thm:fei} still gives no approximation guarantee better than $0$ for the ex-post allocation, because there might be some items with value close to $1/4$ in the bundle of agents. To prove the ex-post guarantee, we provide a more intricate analysis of the method that is used in Theorem \ref{thm:1} and show that our allocation is $1/8$-$\MMS$ ex-post. The fact that the allocation is half-integral plays a key role in the proof. The pseudocode of our approach is shown in Algorithm \ref{algo:rand}.

\begin{restatable}{theorem}{ExAnteExPost}\label{thm:2}
For any instance with $\XOS$ valuations, Algorithm \ref{algo:rand} returns a randomized allocation that is $1/4$-$\MMS$ ex-ante and $1/8$-$\MMS$ ex-post.
\end{restatable}

\subsection*{Deterministic Allocations} Since the impossibility result on $\MMS$ by Proccacia and Wang \cite{10.1145/2600057.2602835}, there has been considerable work of establishing approximate $\MMS$ guarantees for various classes of valuation functions~\cite{Akrami2023BreakingT3,akrami2023simplification,DBLP:journals/teco/BarmanK20, DBLP:journals/ai/GargT21, ghodsi2018fair, kurokawa2018fair, seddighin2022improved}. The best previous deterministic guarantee for $\XOS$ valuations was $0.219225$~\cite{seddighin2022improved}. We improve the guarantee to $3/13 \approx 0.2307$. In order to do so, like in our randomized algorithm, we first allocate large items. However, in addition to allocating single large items, here we also allocate pairs and triples of items if they satisfy some agent up to $3/13$ factor of their $\MMS$ value. This way we get a stronger upper bound on the value of most of the remaining items. 

In the last step, we output an allocation which maximizes the social welfare with respect to valuations $\bar{v}_i(\cdot) = \min\{\frac{6}{13}, v_i(\cdot)\}$ for the remaining agents and items. Note that this last step is also in correspondence with the last step of the randomized algorithm. Here we cap the value of the bundles with $6/13$ (instead of $1/2$) and we output an integral allocation with maximum social welfare (instead of a half-integral allocation with maximum social welfare). Moreover, here we need to do a more careful analysis to show that the output is indeed a $3/13$-$\MMS$ allocation. The pseudocode of our approach is shown in Algorithm \ref{algo:det}.

\begin{restatable}{theorem}{ExPost}\label{thm:ExPost}
For any instance with $\XOS$ valuations, Algorithm \ref{algo:det} returns a $3/13$-$\MMS$ allocation.
\end{restatable}

\subsection{Further Related Work}

\paragraph*{Maximin-share fairness notion.}
Budish~\cite{Budish:first} introduced the $\MMS$-concept. As already mentioned, there are instances with additive valuations for which no $\MMS$ allocation exists~\cite{ feige2022tight, kurokawa2018fair, 10.1145/2600057.2602835}.
Therefore, several studies considered the approximation guarantees. For additive valuations,  the first approximation guarantee, $2/3$-$\MMS$, is given by Kurokawa et al.~\cite{kurokawa2018fair}.
The approximation factor improved over time by Ghodsi et al.~\cite{ghodsi2018fair} to $3/4$, and then to $3/4 + o(1)$ by Garg and Taki~\cite{DBLP:journals/ai/GargT21} and Akrami et al.~\cite{akrami2023simplification}. The best-known result for this valuation class is $3/4 + 3/3836$ which was very recently obtained by Akrami and Garg~\cite{Akrami2023BreakingT3}.
Moreover, there are several works on $\MMS$ approximations for submodular valuations~\cite{DBLP:journals/teco/BarmanK20, ghodsi2018fair}, and subadditive valuations~\cite{ghodsi2018fair, seddighin2022improved}. Also, Farhadi et al.~\cite{asymmetric} proved the tight bound of $1/n$-MMS for the setting where agents can have different entitlements.

\paragraph*{Fractionally subadditive valuations.}

For fractionally subadditive valuations, Ghodsi et al.~\cite{ghodsi2018fair} gave a $1/5$-$\MMS$ approximation algorithm; Seddighin et al.~\cite{seddighin2022improved} improved the factor to $1/4.6 = 0.2173913$. 
A special case of fractionally subadditive valuations has been studied by Li and Vetta~\cite{DBLP:conf/wine/LiV18}.
Also other notions of fairness have been studied for this class of valuation.
Feige~\cite{DBLP:journals/siamcomp/Feige09} studied maximizing social welfare. Also, Hoefer et al.~\cite{DBLP:journals/corr/abs-2209-03908} proved the existence of randomized allocations which are proportional ex-ante and proportional up to one item (PROP1) ex-post.
Furthermore, the notion of envy-freeness up to any item (EFX) has been considered by Plaut and Roughgarden~\cite{DBLP:journals/siamdm/PlautR20}.

\paragraph*{Best-of-both-worlds fairness.}
Aleksandrov et al.~\cite{DBLP:conf/ijcai/AleksandrovAGW15} first explored this line of research in the context of the food bank problem for a special case of additive valuations. 
Freeman et al.~\cite{freeman2020best} proposed a randomized polynomial-time algorithm for additive valuations that is envy-free ex-ante and EF1 ex-post.
Afterward, Aziz~\cite{DBLP:journals/corr/abs-2002-10171} modified this algorithm to also get the Probabilistic Serial fractional outcome~\cite{DBLP:journals/jet/BogomolnaiaM01} as well as a weak notion of efficiency with a simpler proof.
Babaioff et al.~\cite{babaioff2022best}  provided an allocation algorithm in which the expected value of each agent’s bundle is at least a $1/n$-fraction of her value for the set of all items (ex-ante proportional) and ex-post proportional up to one item and $1/2$-MMS allocations for the additive valuations.
Feldman et al.~\cite{cycle-breaking}  gave an allocation algorithm for subadditive valuations with guarantees of $1/2$-envy-freeness ex-ante and $1/2$-EFX and EF1 ex-post.
Moreover, studies have been conducted for additive valuations with binary marginals~\cite{DBLP:journals/corr/abs-2002-10171, DBLP:conf/wine/0002PP020}, and for matroid rank valuations~\cite{aamas23Aziz, DBLP:conf/aaai/BabaioffEF21}.


\section{Preliminaries}\label{sec:prelim}
A fair division instance is denoted by $\mathcal{I} = (\agents, \items, \mathcal{V})$ where $\agents = [n]$ is a set of $n$ agents, $\items$ is a set of $m$ indivisible items and $\mathcal{V}=(v_1, \ldots, v_n)$ is a vector of valuation functions. Each agent $i$ has a valuation function $\valu_i: 2^\items \rightarrow \mathbb{R}_{\geq 0}$ that represents her value for every bundle of items. Thus, for every set $S$ of items, $\valu_i(S)$ shows the value of $S$ to agent $i$.
We denote the $j^{\textrm{th}}$ item by $\ite_j$ and write $v_i(b_j)$ for the value $v_i(\{b_j\})$ of $b_j$ to agent $i$. An allocation of items is a partition of $\items$ into $n$ parts (i.e., bundles) 
where the $i^{\textrm{th}}$ bundle is the share allocated to agent $i$. For an allocation $\allocation$, we denote the bundle of agent $i$ by $\allocation_i$. 

We assume that the valuations are monotone, i.e., for every two sets $S$ and $T$ such that $S \subseteq T$, we have $\valu_i(S) \leq \valu_i(T)$, and normalized, i.e., $\valu_i(\emptyset)=0$ for all $i \in \agents$. Our discussion in this paper involves two classes of valuation functions: additive and fractionally subadditive valuation functions. A valuation function $\valu_i(\cdot)$ is additive, if for every set $S$ of items, we have $$\valu_i(S) = \sum_{\ite_j \in S} \valu_i(\ite_j).$$

\begin{definition}[$\XOS$]\label{def:XOS}
	A valuation function $\valu_i(\cdot)$ is fractionally subadditive ($\XOS$),  if there exists a family of additive valuation functions $u_{i,1},u_{i,2}, \ldots, u_{i,\ell} : 2^\items \rightarrow \mathbb{R}_{\geq 0}$ such that for every set $S$ we have 
	$$
	\valu_i(S) = \max_{1 \leq k \leq \ell} u_{i,k}(S).
	$$
\end{definition}

Given an allocation $\allocation$, we denote by $u_{i,i'}$, an additive function of $\valu_i$ that defines $\valu_{i}(\allocation_i)$, i.e., $\valu_{i}(\allocation_i) = u_{i,i'}(\allocation_i)$.
Another term we frequently use in this paper is contribution; which is defined as the marginal value of one set to another.

\begin{definition}[Contribution]\label{def:contribution}
    For every sets $S, T$ of items such that $S \subseteq T$, we define the marginal contribution of $S$ to $T$ with respect to valuation function $v$, denoted by $C_v^T(S)$, as follows:

$$
	C_v^T(S) = v(T) - v(T \setminus S)
$$
i.e., the marginal contribution of $S$ to $T$ is the value decrease when $S$ is removed from $T$.
\end{definition}

\begin{example}
	Consider 5 items $\ite_1, \ite_2,\ldots, \ite_5$, and an identical valuation $\valu$ for all the agents. Suppose that $\valu$ is a fractionally subadditive function consisting of two additive function $u_1 = [2,8,4,5,1]$ and $u_2 = [5,1,9,4,5]$ 
 (the $j^{\textrm{th}}$ element is the value for $\ite_j$). For set $S= \{\ite_1,\ite_2,\ldots,\ite_5\}$ we have $u_1(S) = 20$ and $u_2(S) = 24$. Hence, $\valu(S) = \max(u_1(S),u_2(S)) = 24$. Also, the marginal contribution of item $\ite_3$ to set $S$ is 
	$C_v^S(\{\ite_3\}) = \valu(S) - \valu(S \setminus \{\ite_3\}) = 24 - 16 = 8$, which is smaller than $u_2(\ite_3) = 9$.
\end{example}
With abuse of notation, for an allocation $\allocation$ of items to agents with valuation vector $\mathcal{V}=(v_1, \ldots, v_n)$ and every set $S$ of items, we define the contribution of $S$ to $\allocation$ with respect to $\mathcal{V}$, denoted by $C_\mathcal{V}^\allocation(S)$ as follows:
\begin{equation}\label{eq:3}
	C_\mathcal{V}^{\allocation}(S) = \sum_{1\leq i \leq n} C_{v_i}^{\allocation_i}(\allocation_i \cap S). 
\end{equation}

However, since an $\XOS$ valuation function might include many additive functions,  Equality \eqref{eq:3} is not always practical. Therefore, we use Observation \ref{eq:2} to bound $C_\mathcal{V}^{\allocation}(S)$. 
Since $u_{i,i'}(\allocation_i) = u_{i,i'}(\allocation_i \setminus S) + u_{i,i'}(\allocation_i \cap S)$ and we have $v_i(\allocation_i) - v_i(\allocation_i \setminus S) \le  u_{i,i'}(\allocation_i) - u_{i,i'}(\allocation_i \setminus S) = u_{i,i'}(\allocation_i \cap S)$. Summing over $i$ yields the following Observation.

\begin{observation}\label{eq:2}
For every allocation $\allocation$ of items to agents with valuation vector $\mathcal{V}$ and every set $S$ of items, we have
	$$C_\mathcal{V}^{\allocation}(S) \leq \sum_{1 \leq i \leq n} u_{i,i'}(\allocation_i \cap S).$$
\end{observation}

Our goal is to allocate the items to the agents in a fair manner.  Here we discuss two share-based notions of fairness, namely proportionality, and maximin-share.
For agent $i$, we define the proportional share of agent $i$, denoted by $\pi_i$ as:
$$
\pi_i = \valu_i(\items)/n.
$$
We also define the maximin-share of agent $i$, denoted by $\MMS_i$ as 
$$
\MMS_i = \max_{\langle \Pi_1,\Pi_2,\ldots, \Pi_n\rangle \in \Omega}\, \min_{1 \leq j \leq n} \valu_i(\Pi_j),
$$
where $\Omega$ is the set of all partitions of $\items$ into $n$ bundles. 
Moreover, if for a partition $\Pi = \langle \Pi_1,\Pi_2,\ldots, \Pi_n\rangle$ of $\items$ into $n$ bundles we have $v_i(\Pi_j) \geq \MMS_i$ for all $j \in [n]$, we say $\Pi$ is an ``$\MMS$ partition'' or ``optimal partition'' of agent $i$. By definition of $\MMS$, for every agent $i$ there exists at least one partition $\Pi$ which is an $\MMS$ partition for agent $i$. For brevity, in the rest of the paper, we assume that the valuations are scaled so that for every agent $i$, we have $\MMS_i = 1$. 
We also define approximate versions of these two notions as follows. For any $\alpha>0$, we say an allocation is $\alpha$-proportional, if it guarantees to each agent $i$ a bundle with value at least $\alpha \pi_i$.  Likewise, in an $\alpha$-$\MMS$ allocation, the value of the share allocated to each agent $i$ is at least $\alpha$. 

\paragraph*{Randomized allocation.} In this paper, we also consider randomized allocations. A randomized allocation is a distribution over a set of deterministic allocations. For a randomized allocation $\ralloc$, we denote by $\dom(\ralloc)$ the set of allocations in the support of $\ralloc$. 
For a randomized allocation $\ralloc$, the expected welfare of agent $i$ is defined as 
$$
	\valu_i(\ralloc) = \sum_{\allocation \in \dom(\ralloc)} \valu_i(\allocation_i) \cdot p_\allocation,
$$
where $p_\allocation$ is the probability of allocation $\allocation$ in $\ralloc$.

\paragraph*{Fractional allocation.}En route to proving our results, we leverage another relaxed form of allocation called fractional allocation. In a fractional allocation, we ignore the indivisibility assumption and treat each item as a divisible one. Formally, a fractional allocation $\falloc$ is a set of $nm$ variables $f_{i,j}$ indicating the fraction of item $\ite_j$ allocated to agent $i$. Therefore, we expect allocation $\falloc$ to satisfy the following constraints:

\begin{align*}
	&\forall_{j}  \qquad \sum_{1\leq i \leq n} f_{i,j} \leq 1 &\mbox{(we have one unit of each item)}\\
	&\forall_{i,j} \qquad 0 \leq f_{i,j}  &\mbox{(each agent receives a non-negative share of each item)} 
\end{align*}
A fractional allocation is \emph{complete} if $\sum_i f_{i,j} = 1$ for all items $\ite_j$, i.e., all items are completely allocated.
Given a fractional allocation $\falloc$, we define the utility of agent $i$ for $\falloc$ in the same way as we calculate it for integral allocations: 
$$
\valu_{i}(\falloc) = \max_k \sum_{1 \leq j \leq m}u_{i,k}(\ite_j)f_{i,j}.
$$
Complete fractional allocations give rise to randomized allocations in the standard way, i.e., the probability $p_\allocation$ of an allocation $\allocation$ is defined as
\begin{equation}\label{frac-to-rand} p_\allocation = \prod_{1\leq i \leq n} \prod_{\ite_j \in \allocation_i} f_{ij}. \end{equation}
Then $\sum_\allocation p_\allocation = 1$. Indeed, view an integral allocation $\allocation$ as a mapping $\pi$ from $[m]$ to $[n]$: $\pi(j) = i$ iff $\ite_j \in \allocation_i$. Then $\sum_\allocation p_\allocation = \sum_\allocation \prod_i \prod_{\ite_j \in \allocation_i} f_{ij} = \sum_{\pi \in [n]^{[m]}} \prod_j f_{\pi(j) j} = \prod_j (\sum_i f_{ij}) = \prod_j 1 = 1$.

\begin{lemma}\label{randomized versus fractional} Let $\falloc$ be a complete fractional allocation and let randomized allocation $\ralloc$ be defined by (\ref{frac-to-rand}). Then for $\XOS$ valuation functions $v_i$,
  \[     v_i(\ralloc) \ge v_i(\falloc) \]
  for all $1 \leq i \leq n$. 
\end{lemma}
\begin{proof} For each $i$, let $i'$ be such that $u_{i,i'}(\falloc) = \max_k u_{i,k}(\falloc)$. We view again an allocation $\allocation$ as a mapping $\pi \in [n]^{[m]}$. Then 
  \begin{align*}
    v_i(\ralloc) &= \sum_\allocation v_i(\allocation_i) p_\allocation \\
                        &= \sum_\pi v_i(\pi^{-1}(i)) \prod_t f_{\pi(t),t} \\
                        &\ge \sum_\pi \sum_{j;\ \pi(j) = i} u_{i,i'}(j) \prod_t f_{\pi(t),t}\\
                        &= \sum_j \sum_{\pi;\ \pi(j) = i} u_{i,i'}(j) f_{i,j} \prod_{t \not= j}  f_{\pi(t),t}\\
                        &= \sum_j  u_{i,i'}(j) f_{i,j}  \sum_{\pi \in [n]^{[m] \setminus j}} \prod_{t \not= j} f_{\pi(t),t}\\
                        &= \sum_j  u_{i,i'}(j) f_{i,j}  \prod_{t \not= j} \sum_\ell f_{\ell,t}\\
                        &= \sum_j  u_{i,i'}(j) f_{i,j}  \prod_{t \not= j} 1 \\
                        &= \sum_j  u_{i,i'}(j) f_{i,j} \\
                        &= v_i(\falloc).
  \end{align*}\end{proof}

We also need to redefine \emph{contribution} for fractional allocations and fractional bundles. Suppose that $\falloc$ is a fractional allocation and $S$ is a fractional set of items.
Since items are fractionally allocated, the term contribution must be defined more precisely. For example, suppose that set $S$ consists of a fraction $0.4$ of item $\ite_j$, and in allocation $\falloc$,  $0.2$ of $\ite_j$ belongs to agent $i_1$, $0.5$ of $\ite_{j}$ belongs to agent  $i_2$, and $0.3$ of $\ite_j$ belongs to agent $i_3$. We need to define exactly how the $0.4$ fraction of item $\ite_j$ in $S$ is distributed over the agents. One reasonable strategy is to choose the share of each agent in a way that after removal of $S$ from $\falloc$ we have the smallest possible decrease in the social welfare. Based on this strategy, assuming that $s_{j}$ is the fraction of item $\ite_{j}$ in $S$, we define the contribution of $S$ to $\falloc$, denoted by $C_\mathcal{V}^{\falloc}(S)$ as the answer of the following optimization program:
\begin{align}
	\text{minimize }   &\sum_{1 \leq i \leq n}    \valu_i(\falloc) - \valu_i(\falloc')          \nonumber  \\
	\text{subject to } &  \sum_{1 \leq i \leq n} f_{i,j} -f'_{i,j} = s_{j}& \forall_j \nonumber \\
	&0 \leq  f'_{i,j} \leq f_{i,j}& \forall_{i,j} \label{eq:contr} 
\end{align}
Generally, it is hard to deal with the above optimization program. Here, we use an important property of $C^{\falloc}_\mathcal{V}(\cdot)$ to obtain our results. 
\begin{lemma}\label{lem:supper-additive}
	Let $\falloc$ be an arbitrary fractional allocation and assume that for every agent $i$, $\valu_i(\cdot)$ is $\XOS$. Then, for every partition of the items into fractional sets $S_1,S_2,\ldots,S_t$, we have $$\sum_{1 \leq k \leq t} C_\mathcal{V}^{\falloc}(S_k) \leq \sum_{1 \leq i \leq n} \valu_i(\falloc).$$
\end{lemma}
\begin{proof} 
For every agent $i$, let $i' =\arg \max_k u_{i,k}(\falloc)$. 
By definition, we have
	$$\valu_{i}(\falloc) = \sum_{1 \leq j \leq m} u_{i,i'}(\ite_j)\cdot f_{i,j}.$$
 
We will define for each set $S_k$ an allocation $\falloc^{(k)}$ by reducing the allocation $\falloc$ proportionally, i.e., we will replace $f_{ij}$ by $(1 - s_{k,j})f_{ij}$, where $s_{k,j}$ is the fraction of item $j$ belonging to set $S_k$. Note that since $S_1,S_2,\ldots,S_t$ is a partition of items, $\sum_k s_{k,j} = 1$. For every $1 \leq k \leq t$, $1\leq i \leq n$, and $1 \leq j\leq m$, define variable $f'^{(k)}_{i,j}$ as follows:
	$$
	f'^{(k)}_{i,j}= (1-s_{k,j}) \cdot f_{i,j}.
	$$
	Denote by $\falloc'^{(k)}$ the partial allocation defined by the variables $f'^{(k)}_{i,j}$. Since for every $j$, we have
	\begin{align*}
		\sum_{1 \leq i \leq n} (f_{i,j} - f'^{(k)}_{i,j}) &= \sum_{1 \leq i \leq n} s_{k,j}f_{i,j}\\ &= s_{k,j},
	\end{align*}
	$\falloc'^{(k)}$ is a feasible solution to Program \eqref{eq:contr}. Therefore, we have
	\begin{align*}
		\sum_{1 \leq k \leq t} C^{\falloc}_\mathcal{V}(S_k) &\le \sum_{1 \leq k \leq t} \sum_{1 \leq i \leq n} (v_i(\falloc) - v_i(\falloc^{'(k)})) \\
		&\le \sum_{1 \leq k \leq t} \sum_{1 \leq i \leq n} \sum_{1 \leq j \leq m} u_{i, i'}(\ite_j) (f_{i,j} - f_{i,j}^{'(k)}) \\
		& =  \sum_{1 \leq i \leq n} \sum_{1 \leq j \leq m}\sum_{1 \leq k \leq t} u_{i, i'}(\ite_j) s_{k,j} f_{i,j} \\
		&=  \sum_{1 \leq i \leq n} \sum_{1 \leq j \leq m} u_{i, i'}(\ite_j) f_{i,j}\\
		&= \sum_{1 \leq i \leq n} \valu_i(\falloc).
	\end{align*}
\end{proof}
\subsection{Ex-ante and Ex-post Fairness Guarantees.}
For a randomized allocation, we define two types of fairness guarantees, namely \emph{ex-ante} and \emph{ex-post} as in Definitions \ref{def:2} and \ref{def:3}. 
\begin{definition}[ex-ante]\label{def:2}
	Given a randomized allocation $\ralloc$, we say $\ralloc$ is $\alpha$-$\MMS$ \emph{ex-ante}, if for every agent $i$, we have $\valu_i(\ralloc)\geq \alpha $. Similarly, $\ralloc$ is $\alpha$-proportional, if for every agent $i$, $\valu_i(\ralloc)\geq \alpha \pi_i$. 
\end{definition}

\begin{definition}[ex-post]\label{def:3}
	An allocation $\ralloc$ is $\alpha$-$\MMS$ \emph{ex-post}, if every allocation $\allocation \in \domain(\ralloc)$ is $\alpha$-$\MMS$. Similarly, we say $\ralloc$ is $\alpha$-proportional ex-post if every allocation $\allocation \in \domain(\ralloc)$ is $\alpha$-proportional.
\end{definition}

One tool that we refer to in this paper is the result of Babaioff, Ezra, and Feige for converting a fractional allocation into a faithful randomized allocation \cite{babaioff2022best}. 

\begin{theorem}[Proved in \cite{babaioff2022best}]\label{thm:fei}
	Assume that the valuations are additive and let $\falloc$ be a fractional allocation. Then there exists a randomized allocation 
	$\ralloc$ such that the ex-ante utility of the agents  for $\ralloc$ is the same as the utility of the agents in $\falloc$, and for every allocation $\allocation$ in the support of $\ralloc$ the following holds: 
	$$
	\forall_i, \valu_i(\allocation_i) \geq \valu_i(\ralloc) - \max_{j: f_{i,j} \notin \{0,1\}} \valu_i(\ite_j).
	$$
\end{theorem}

In this paper, we use a more delicate analysis of the method used in Theorem \ref{thm:fei} to convert fractional allocations into randomized ones. This helps us improve our ex-post approximation guarantee.
\section{Ex-ante Guarantees}\label{sec:exante}
In this section, our goal is to explore the possibility of designing randomized allocations that is $\alpha$-$\MMS$ ex-ante or $\alpha$-proportional ex-ante.  Note that, in contrast to the additive case, for $\XOS$ valuations there is no meaningful correspondence between proportionality and maximin-share; proportional share can be larger or smaller than maximin-share. Recall that for the additive case, we always have $\pi_i\geq \MMS_i$ and therefore, maximin-share is implied by proportionality. However, for fractionally subadditive valuations, $\pi_i$ can be as small as $\MMS_i/n$. 

For the additive setting, a simple fractional allocation that allocates a fraction $1/n$ of each item to each agent guarantees proportionality and consequently maximin-share. Using Theorem \ref{thm:fei} one can convert this allocation to a randomized allocation that is proportional ex-ante. In Observation $\ref{obs:3}$ we show that proportionality can be guaranteed ex-ante for $\XOS$ valuations.\footnote{Note that for now, we are not concerned about the ex-post guarantee of our allocation.}

\begin{observation}\label{obs:3}
	Every randomized allocation that allocates each item with probability $1/n$ to each agent is proportional ex-ante.
\end{observation}
\begin{proof}
	Let $\ralloc$ be a randomized allocation such that the probability that item $\ite_j$ is allocated to agent $i$ is $1/n$. Also, let $u_{i,i'}$ be the additive valuation function that defines $\valu_i(\items)$, i.e., $v_i(\items) = u_{i,i'}(\items)$, and let $p_\allocation$ be the probability that allocation $\allocation$ is chosen in $\ralloc$. We have 
	\begin{align*}
		\valu_{i}(\ralloc) &= \sum_{\allocation \in \domain(\ralloc)} p_\allocation \valu_{i}(\allocation_i)\\
		&\geq \sum_{\allocation \in \domain(\ralloc)} p_\allocation u_{i,i'}(\allocation_i)\\		&= \sum_{\allocation \in \domain(\ralloc)} p_\allocation\sum_{\ite_j \in \allocation_i} u_{i,i'}(\ite_j)\\
		&= \sum_{1 \leq j \leq m}\sum_{\allocation_i \ni \ite_j} u_{i,i'}(\ite_j)p_\allocation\\
		&= \sum_{1 \leq j \leq m} u_{i,i'}(\ite_j)/n\\
		&=\valu_i(\items)/n.
	\end{align*}
\end{proof}

In contrast to the additive setting, finding an allocation that guarantees maximin-share is not trivial. Indeed, the simple fractional allocation that guarantees proportionality in Observation \ref{obs:3} can be as bad as $O(1/n)$-$\MMS$. 

\begin{observation}\label{obs:2}
	Let $\falloc$ be a fractional allocation that allocates a fraction $1/n$ of each item to each agent. Then, there exists an instance such that the maximin-share guarantee of $\falloc$ is $O(1/n)$.  
\end{observation}
\begin{proof}
	Consider the following instance: there are $n^2$ items. The valuation of agent $i$ is an $\XOS$ set function consisting of $n$ additive valuation functions as follows: partition the items into $n$ bundles each with $n$ items. For each additive function $u_{i,k}$, the value of each item in the $k^{\textrm{th}}$ bundle is $1/n$ and the value of the rest of the items is $0$. It is easy to observe that for this instance, the $\MMS$ value of each agent is $1$, and the value of each agent for her bundle in $\falloc$ is $1/n$.
\end{proof}

Generally, there are two main challenges in the process of designing a randomized allocation that guarantees an approximation of maximin-share. In contrast to the additive setting, finding a fractional or randomized allocation that approximates maximin-share is not easy. As well as that, transforming a fractional allocation into a randomized one is not straightforward.
Indeed, as we show in Lemma \ref{lem:1}, neither fractional allocations nor randomized allocations can guarantee $\MMS$. We prove an upper bound on the best approximation guarantee of each one of these allocation types. 

\begin{lemma}\label{lem:1}
	For $\XOS$ valuations, the best $\MMS$ guarantee for fractional allocations and the best ex-ante $\MMS$ guarantee for randomized allocation is upper bounded by $3/4$.
\end{lemma}
\begin{proof}
	Consider the following instance: there are two agents and four items. The fractionally subadditive valuation of each agent consists of 2 additive functions. The valuations are as follows:
	\begin{align*}
		u_{1,1}(\{b_1\}) = u_{1,1}(\{b_2\}) = 1,u_{1,1}(\{b_3\}) = u_{1,1}(\{b_4\}) = 0,\\
		u_{1,2}(\{b_1\}) = u_{1,2}(\{b_2\}) = 0,u_{1,2}(\{b_3\}) = u_{1,2}(\{b_4\}) = 1,\\
		u_{2,1}(\{b_1\}) = u_{2,1}(\{b_4\}) = 1,u_{2,1}(\{b_2\}) = u_{2,1}(\{b_3\}) = 0,\\
		u_{2,2}(\{b_1\}) = u_{2,2}(\{b_4\}) =0 ,u_{2,2}(\{b_2\}) = u_{2,2}(\{b_3\}) = 1.
	\end{align*}
	It is easy to check that for the above instance, the maximin-share of each agent is equal to $2$, and no fractional allocation can guarantee more than $1.5$ to both agents. Also, we can guarantee a value of $1.5$ to both agents by giving the first and the third items receptively to agents 1 and 2, and giving half of the remaining items to each agent. 
	
	Now, we show that the same upper bound also holds for the ex-ante guarantee of randomized allocations.  Assume that $\ralloc$ is the randomized allocation that maximizes the maximin-share guarantee for this instance. Since there are two agents, we know that $\ralloc$ maximizes the following objective:
	$$
	\alpha = \min \left(\sum_{S \subseteq \items} \mathbb{P}(S) \valu_1(S),\sum_{S \subseteq \items} \mathbb{P}(S) \valu_2(\items \setminus S)\right),
	$$
	where $\mathbb{P}(S)$ is the probability that set $S$ is allocated to agent $1$. Furthermore, since for every integers $y,z$ we have $\min(y,z) \leq (y+z)/2$, and
	\begin{align*}
		\alpha &\leq \left(\sum_{S \subseteq \items} \mathbb{P}(S) \valu_1(S)+\sum_{S \subseteq \items} \mathbb{P}(S) \valu_2(\items \setminus S)\right)/{2}\\
		&=\sum_{S \subseteq \items} \mathbb{P}(S) \left(\valu_1(S)+\valu_2(\items \setminus S)\right) /2.
	\end{align*}
	One can easily check that for every set $S$, the value of $\valu_1(S)+\valu_2(\items \setminus S)$ is upper bounded by $3$. Therefore, we have 
	\begin{align*}
		\alpha &\leq \sum_{S \subseteq \items} (3/2) \mathbb{P}(S)\\
		&\leq 3/2. 
	\end{align*}
	Hence, the best possible approximation guarantee for $\MMS$ in this instance is at most $1.5/2 = 3/4$. Note that one can guarantee a value of $1.5$ ex-ante to both the agents by giving the first and the third items receptively to agents 1 and 2, and giving the rest of the items with a probability of $1/2$ to one of the agents.
\end{proof}

Before we prove our lower bound on the maximin-share guarantee for randomized allocations, we note that another challenge about $\XOS$ valuations is that in sharp contrast to additive valuations, transforming a fractional allocation to a randomized one is not easy. Indeed, we can show that for a fractional allocation $\falloc$ there might be randomized allocations $\ralloc$ and $\ralloc'$ with different utility guarantees for the agents, such that in both $\ralloc$ and $\ralloc'$ the probability that each item $\ite_{j}$  is allocated to agent $i$ is equal to $f_{i,j}$. Example \ref{ex:1} gives more insight into this challenge. 

\begin{example} \label{ex:1}
	Consider the instance described in the proof of Observation \ref{obs:2} and define allocations $\ralloc$ and $\ralloc'$ as follows:
	\begin{itemize}
		\item Allocation $\ralloc$ allocates each item to each agent with probability $1/n$. 
		\item Allocation $\ralloc'$ considers a random permutation of the bundles in the optimal $\MMS$-partitioning of   the agents and allocates the $i^{\textrm{th}}$ bundle in the permutation to agent $i$.\footnote{Note that in the instance described in Observation \ref{obs:2} the optimal $\MMS$-partitioning of all the agents are the same.} 
	\end{itemize}
	It is easy to check that in both of these allocations, each item is allocated to each agent with probability $1/n$. However, the maximin-share guarantee of $\ralloc$ is $O(\log n /n)$. To show this, one can argue that using Chernoff bound the probability that more than $3\log n$ items from the same bundle in the optimal partition are allocated to agent $i$ is $O(1/n^2)$. Hence, the expected value of agent $i$ for her share is at most
	$$
	1 \cdot \frac{1}{n^2}+ \frac{3\log n}{n} \cdot \frac{n^2-1}{n^2} \leq \frac{4 \log n}{n}.
	$$  
	
	On the other hand, allocation $\ralloc'$ guarantees value $1$ to all the agents. 
\end{example}

Despite these hurdles, in Theorem \ref{thm:1} we show that there exists a randomized allocation that guarantees $1/4$-$\MMS$ to all the agents ex-ante. To prove Theorem \ref{thm:1}, we first show that a fractional allocation exists that is $1/4$-$\MMS$. Next, we convert it to a randomized allocation. Theorem \ref{thm:1} along with Lemma \ref{lem:1} leave a gap of $[1/4,3/4)$ between the best upper bound and the best lower-bound for the maximin-share guarantee of fractional allocations in the $\XOS$ setting.

\exante*
    {\color{black} \textbf{Proof Idea:} Let $\falloc$ be the a (fractional) allocation maximizing social welfare and assume that there is an agent $i$ whose bundle has value less than $1/4$ to her. Now consider the $\MMS$-partition of agent $i$. Each bundle in this partition has value at least 1 to $i$. We can split each bundle into two so that each subbundle has value at least $1/2$ to $i$. Let $B$ be one of these $2n$ subbundles. Imagine that we reassign the items in $B$. We take away the items in $B$ from their current owners and give them to $i$. Then $i$ would gain more than $1/4$, but the other agents would lose. The loss is bounded by $C_{\falloc}(B)$. Why should this quantity be less than $1/4$ for one of the $2n$ subbundles?

Lemma~\ref{lem:supper-additive} comes to the rescue. We have 
\[ \sum_{1 \leq j \leq 2n} C_\mathcal{V}^{\falloc}(B_j) \leq \sum_{1 \leq i \leq n} \valu_i(\falloc),\]
where $B_1$ to $B_{2n}$ are the subbundles. If the right hand side is strictly less than $n/2$, the desired subbundle exists. We can achieve this by replacing our valuations $\valu_i$ by valuations $\bar{\valu}_i$ that assign no set a value more than $1/2$.}

\begin{proof} For a fractional allocation, we define the truncated value of agent $i$, denoted by  $\overline{\valu}_i$, of a fractional set $S$ as follows:
	\begin{equation}\label{vbar}
		\overline{\valu}_i(S) = \min\bigg(\frac{1}{2}, \max_k \sum_{1 \leq j \leq m} u_{i,k}(\ite_j)s_{j} \bigg).
	\end{equation}
	where $u_{i,k}$ is the $k^{\textrm{th}}$ additive function of $\valu_i$ and $s_j$ is the fraction of item $\ite_j$ that belongs to set $S$. If the valuation $\valu_{i}(\cdot)$ is $\XOS$, then $\bar{\valu}_i(\cdot)$ is also $\XOS$~\cite{ghodsi2018fair}. Let $\bar{v}=(\bar{v}_1, \ldots, \bar{v}_n)$.
	
	Now let $\falloc$ be the complete fractional allocation that maximizes 
	\begin{equation}\label{truncsocial}
		Z=\sum_{1 \leq i \leq n} \bar{\valu}_i(\falloc).
              \end{equation}
        
	We know $Z \leq n/2$ since for any fractional bundle $S$ and every agent $i$, we have $ \overline{\valu}_i(S) \leq 1/2$. We claim that $\falloc$ allocates each agent a bundle with a value of at least $1/4$. For the sake of a contradiction, assume that this is not true and let agent $a$ be an agent whose share is worth less than $1/4$ to her. Then $Z <n/2-1/4$.
	
	Given that the maximin-share of agent $i$ is equal to $1$, she can divide the items in $\items$ into $2n$ fractional bundles, each of which has a value of at least $1/2$ to her.
	
	 \begin{claim}\label{clm:divide}
	 	Since the maximin-share of agent $i$ is at least $1$, she can divide the item in $\items$ into $2n$ fractional bundles, each with value at least $1/2$ to her.
	 \end{claim}
	\begin{claimproof}
	Consider the optimal maximin-share partition of agent $i$, and for each bundle in this partition divide that bundle into two fractional sub-bundles with a value of at least $1/2$. Since the valuation of agent $i$ is $\XOS$, such a division is always possible: just take a fractional sub-bundle with a value of exactly $1/2$ from each bundle. The remaining (fractional) items in that bundle also form a sub-bundle with a value of at least $1/2$. 
	\end{claimproof}
	
	Let ${B}_1,{B}_2,\ldots,{B}_{2n}$ be these $2n$ bundles. By applying Lemma~\ref{lem:supper-additive} we have 
	$$
	\sum_{1 \leq j \leq 2n} C_{\bar{v}}^{\falloc}(B_j) \leq \sum_{1 \leq i \leq n} \bar{\valu}_i(\falloc) = Z.
	$$
	Note that here $C_{\bar{v}}^{\falloc}$ refers to the contribution with respect to $\bar{v}_i$.	Therefore, at least one of the bundles contributes less than $Z/(2n) < 1/4$ to $Z$.
	Let $B_k$ be one such bundle, i.e., $C_{\bar{v}}^{\falloc}(B_k) < 1/4$. Let $b_{kj}$ be the fraction of item $j$ belonging to bundle $B_k$ and let $\falloc'$ be the allocation that defines the contribution of bundle $B_k$ to allocation $\falloc$ (see Program~\ref{eq:contr}).  Then $\sum_i f'_{ij} = \sum_i f_{ij} - b_{kj}$ for all $j$ and
        {\[\sum_{1 \leq i \leq n} \bar{\valu}_i(\falloc) -  \sum_{1 \leq i \leq n} \bar{\valu}_i(\falloc') =  C_{\bar{v}}^\falloc(B_k) <1/4,\]}
        which means
 {$$
	\sum_{1 \leq i \leq n} \bar{\valu}_i(\falloc') > Z - \frac{1}{4}.
	$$}
        We now assign the items in $B_k$ to agent $i$, i.e., we consider the fractional allocation $\falloc''$ equal to $\falloc'$, except that for agent $i$, we have 
	$$
	f''_{i,j} = f'_{i,j} + b_{k,j} \qquad \text{for all $j \in [1\ldots m]$}
	$$
        Since the value of bundle $B_k$ to agent $a$ is at least $1/2$, we have {$\bar{\valu}_{i}(\falloc'') - \bar{\valu}_{i}(\falloc') >1/4$}, and further
      \begin{equation}\label{contradiction}
		\sum_{1 \leq i \leq n} \bar{\valu}_i(\falloc'') >   \sum_{1 \leq i \leq n} \bar{\valu}_i(\falloc') + \frac{1}{4} 
		> (Z - \frac{1}{4}) + \frac{1}{4}
		= Z 
		= \sum_{1 \leq i \leq n} \bar{\valu}_i(\falloc).
	\end{equation}
	However, Inequality \eqref{contradiction} contradicts the fact that allocation $\falloc$ maximizes the social welfare. 
	Hence, $\falloc$ guarantees at least $1/4$ to all the agents.
        
        Finally, let $\ralloc$ be the randomized allocation obtained from $\falloc$ through (\ref{frac-to-rand}). Then $v_i(\ralloc) \ge v_i(\falloc) \ge \bar{v}_i(\falloc) \ge 1/4$ by Lemma~\ref{randomized versus fractional}. 
        Thus, $\ralloc$ is $1/4$-$\MMS$ ex-ante. This completes the proof.
\end{proof}

We remark that though we constructed a $1/4$-$\MMS$ allocation, 
we have no guarantee on the ex-post fairness of our allocation. In the next section, our goal is to improve this allocation to also guarantee a fraction of maximin-share ex-post.
\section{Ex-ante and Ex-post Guarantees} 

Unfortunately, the randomized allocation obtained by Theorem \ref{thm:1} has no ex-post fairness guarantee. The issue is that we use  Theorem \ref{thm:fei} to convert the fractional allocation into a randomized one. However, Theorem \ref{thm:fei} only guarantees that the ex-post value of each agent is at least the value of her fractional allocation minus the value of the heaviest item which is partially (and not fully) allocated to her in the fractional allocation. However, currently, we have no upper bound on the value of the allocated items, and therefore, the ex-post value of an agent might be close to $0$. To resolve this, we perform two improvements on our allocation. 

First, we allocate valuable items beforehand to keep the value of the remaining items as small as possible. 
We start by using a simple and very practical fact that is frequently used in previous studies~\cite{ghodsi2018fair,seddighin2022improved,barman2017approximation,amanatidis2019multiple}: allocating one item to one agent and removing them from the instance does not decrease the maximin-share value of the remaining agents for the remaining items. 

\begin{lemma}\label{ass:1}
	Removing one item and one agent from the instance does not decrease the maximin-share value of the remaining agents for the remaining items.
\end{lemma}

Given that our goal is to construct a randomized allocation which is $1/4$-$\MMS$ ex-ante, by Lemma \ref{ass:1} we can assume without loss of generality that the value of each item to each agent is less than $1/4$; otherwise, we can reduce the problem using Lemma \ref{ass:1}.
However, a combination of this assumption and Theorem \ref{thm:1} still gives no ex-post guarantee: the ex-ante guarantee obtained by Theorem \ref{thm:1} is $1/4$-$\MMS$ and assuming that the value of each item to each agent is less than $1/4$ implies no lower-bound better than $0$ on the ex-post $\MMS$ guarantee. 
To improve the ex-post guarantee, we revisit the proof of Theorem \ref{thm:fei} and show that for our setting, a stronger guarantee can be achieved using the matching method for converting a fractional allocation into a randomized one. Indeed, we show that we can find a fractional allocation with a special structure that makes the transformation step more efficient. These ideas together help us achieve a randomized allocation with $1/4$-$\MMS$ guarantee ex-ante and $1/8$-$\MMS$ guarantee ex-post.

\ExAnteExPost*
In the rest of this section, we prove Theorem \ref{thm:2}. The algorithm we use for proving Theorem \ref{thm:2} is as follows:
\begin{enumerate}
    \item \label{step-11} While there exists an item $\ite_{j}$ with value at least $1/4$ to an agent $i$, allocate $\ite_j$ to agent $i$ and remove $i$ and $\ite_j$ respectively from $\agents$ and $\items$. 
    \item \label{step-12} Assuming $[n]$ is the set of the remaining agents, let $\falloc$ be an optimal solution of  the following linear program.
      \begin{align}
	\text{maximize }   &\sum_{1 \leq i \leq n}    u_i  \nonumber      \\
	\text{subject to } &  \sum_{1 \leq i \leq n} f_{i,j} = 1& \text{for all $j$}  \nonumber\\
	&  f_{i,j} \in \{0,1/2,1\} & \text{for all $i$ and $j$} \nonumber \\
	&u_i = \min(\frac{1}{2}, \max_k \sum_j u_{i,k}(\ite_j) f_{i,j}).
	& \text{for all $i$} \label{lp}	
\end{align}
    \item \label{step-13} Convert $\falloc$ into a randomized allocation using Lemma \ref{lem:fractorand}.
\end{enumerate}

\begin{algorithm} [tb]
	\caption{$\mathtt{ExPostExAnteMMS}(\agents, \items, \mathcal{V})$ 
        \\ \textbf{Input:} Instance $(\agents, \items, \mathcal{V})$.
        \\ \textbf{Output:} Allocation $\ralloc$. 
    }
    \label{algo:rand}
	\begin{algorithmic}[1]
        \While{there exists $\ite_j \in \items$ and $i \in \agents$ s.t. $v_i(\ite_j) \geq 1/4$} \Comment{Step \ref{step-11}}
            \State $\ralloc_i \leftarrow \{\ite_j\}$
            \State $\items \leftarrow \items \setminus \{\ite_j\}$
            \State $\agents \leftarrow \agents \setminus \{i\}$
        \EndWhile
        \State Let $\bar{\valu}_i(\cdot) = \min(1/2,\valu_{i}(\cdot))$
        \State Let $\Pi$ be the set of all half-integral allocations of $\items$ to $\agents$
        \State Let $\falloc = \arg \max_{F \in \Pi} \sum_{i \in \agents} \bar{v}_i(F_i)$ \Comment{Step \ref{step-12}}
        \State Let $\ralloc$ be the randomized allocation obtained from $\falloc$ by Lemma \ref{lem:fractorand}. \Comment{Step \ref{step-13}}
        \State Return $\ralloc$
	\end{algorithmic}	
      \end{algorithm}
      
      \noindent
      See Algorithm \ref{algo:rand} for the pseudocode. Recall that by Lemma \ref{ass:1}, after Step \ref{step-11}, the maximin-share value of the remaining agents for the remaining items is at least $1$. For simplicity, we scale the valuations after the first step so that the $\MMS$ value of each remaining agent after the first step is exactly equal to $1$. All agents $i$ who are allocated an item $b_j$ in Step \ref{step-11}, are also allocated $b_j$ in the fractional allocation. Thus, $1/4$-$\MMS$ is guaranteed for $i$ in the final randomized allocation both ex-ante and ex-post. Hence, it is without loss of generality to ignore these agents and assume $n$ is the number of remaining agents after Step \ref{step-11}.

An equivalent description for Step \ref{step-12} is the following. For every agent $i$, define $\bar{\valu}_i$ as follows: 
	$$\forall S \subseteq \items \quad \bar{\valu}_i(S) = \min(1/2,\valu_{i}(S)).$$
	Let $\bar{v} = (\bar{v}_1, \ldots, \bar{v}_n)$ and return a half-integral allocation $\falloc$ that maximizes social welfare with respect to $\bar{v}$, i.e., $\falloc = \arg \max_{A \in \Pi} \sum_{i \in \agents'} \bar{v}_i(A_i)$ where $\Pi$ is the set of all half-integral allocations of $\items$ to $\agents$. 
The goal in Step \ref{step-12} is to find a fractional allocation that is $1/4$-$\MMS$. However, we want this allocation to have a special structure that facilitates constructing the randomized allocation.  Therefore, instead of directly choosing the allocation that maximizes social welfare, we consider $\bar{\valu}_i$ as the valuation function of agent $i$ and return a half-integral allocation. First we prove that $\falloc$ is $1/4$-$\MMS$. 
Otherwise, let $i$ be an agent that has a value less than $1/4$ for her share. By Claim \ref{clm:divide}, we know that agent $i$ can distribute all the items (that have remained after Step \ref{step-11}) into $2n$ bundles each with value at least $1/2$ to her. Here, we construct these $2n$ bundles more carefully.

Indeed, for every bundle in the optimal partitioning of agent $i$, we construct two bundles with a value of at least $1/2$ as follows: we divide each item into two half-unit items and put each half-unit into one bundle. That way, for all items $b_j$, there are two bundles each of which contains one half of $b_j$. 

Using the same deduction as we used in the proof of Theorem \ref{thm:2}, we can say that since the number of remaining agents after Step \ref{step-11} is $n$, the value of one agent for her bundle is less than $1/4$, and the value of the rest of the agents for their bundles is at most $1/2$, the social welfare of the allocation is less than $n/2$. Therefore, at least one of these $2n$ bundles, say $B_k$ contributes less than $1/4$ to social welfare. Now we take back these items from other agents and allocate them to agent $i$. The reallocation increases social welfare as shown in the proof of Theorem \ref{thm:1}. Note that also in this new allocation for every agent $i$ and item $\ite_j$, we have $f_{i,j} \in \{0,1/2,1\}$  which is a contradiction with the choice of $\falloc$. 


\begin{observation}\label{obs:stupper}
	Let $\falloc$ be the allocation after Step \ref{step-12}. Then, for every agent $i$ we have $\valu_i(\falloc) \geq 1/4$. Furthermore, for every item $\ite_j$, we have $f_{i,j} \in \{0,1/2,1\}$. 
\end{observation}
\begin{proof}
    Towards a contradiction assume $v_i(\falloc) < 1/4$ for some agent $i$. Let $B_1, B_2, \ldots, B_{2n}$ be the result of halving the bundles in the optimal partitioning of agent $i$, i.e., dividing each item into two half-unit items and putting each half-unit into one bundle.
    Let $\falloc'$ be the allocation (see Program~\ref{eq:contr}) that defines the contribution of bundle $B_k$ to allocation $\falloc$.  Then $\sum_i f'_{ij} = \sum_i f_{ij} - b_{kj}$ for all $j$ and
{$\sum_{1 \leq i \leq n} \bar{\valu}_i(\falloc) -  \sum_{1 \leq i \leq n} \bar{\valu}_i(\falloc') =  C_{\bar{v}}^\falloc(B_k) <1/4$}. We now assign the items in $B_k$ to agent $i$, i.e., we consider the fractional allocation $\falloc''$ equal to $\falloc'$, except that for agent $i$, we have 
	$$
	\forall_{1 \leq j \leq m} \qquad f''_{i,j} = f'_{i,j} + b_{k,j}.
	$$
        Since the value of bundle $B_k$ to agent $a$ is at least $1/2$, we have {$\bar{\valu}_{i}(\falloc'') - \bar{\valu}_{i}(\falloc') >1/4$}, and hence 
$\sum_{1 \leq i \leq n} \bar{\valu}_i(\falloc'') >  \sum_{1 \leq i \leq n} \bar{\valu}_i(\falloc)$. Since for every agent $i'$ and item $\ite_j$, we have $f''_{i',j} \in \{0,1/2,1\}$ and hence a contradiction to the optimality of $\falloc$. 
\end{proof}

Now, we show how to convert $\falloc$ into a randomized allocation. Recall that the result of Theorem \ref{thm:fei} does not provide us with an ex-post guarantee better than $0$.  
Here, we give a more accurate analysis to prove that the outcome of our algorithm is $1/8$-$\MMS$.  Our construction is based on the Birkhoff---von Neumann theorem: Every fractional perfect matching can be written as a linear combination of integral perfect matchings. We adopt the construction to our setting and, in particular, exploit the fact that all $f_{ij}$ are half-integral.

	
	\begin{lemma} \label{lem:fractorand} Assume that the valuations are additive and let $\falloc$ be a complete fractional allocation with $f_{ij} \in \{0, 1/2, 1\}$ for all $i$ and $j$. Then there is a randomized allocation $\ralloc$ with $\domain(\ralloc) = \{\allocation^1,\allocation^2\}$, such that
		\begin{itemize}
			\item  For every agent $i$ we have $\valu_i(\ralloc) = \valu_{i}(\falloc)$. 
			\item For every agent $i$ we have  
			\[ \min\bigg\{\valu_{i}(\allocation^1_i),\valu_{i}(\allocation^2_i)\bigg\}\geq  v_i(\ralloc) - \frac{\max \{ v_i(b_j)\, | \, f_{ij} = 1/2\}}{2}. \]
		\end{itemize}
	\end{lemma}
	\begin{proof}
		For each agent $i$, let $f_i = \sum_j f_{ij}$.  Since $\sum_i f_i= m$,  the number of agents with non-integral $f_i$ is even. We pair the agents with non-integral $f_i$ arbitrarily. For each pair, we create a new dummy item with a value of zero for all the agents and assign one-half of the dummy item to each agent in the pair. In this way, for every agent $i,$  $f_i$ becomes an integer. Therefore, for the rest of the proof we assume that for every agent $i$, $f_i$ is an integer.
		
		We now construct allocations $\allocation^1$ and $\allocation^2$ such that $f_{ij}$ is equal to the fraction of allocations in $\ralloc$ that allocate item $\ite_j$ to $i$, i.e., if $f_{ij} = 1$ we allocate $\ite_j$ to $i$ in both allocations, if $f_{ij} = 0$, we allocate $\ite_j$ to $i$ in neither allocations, and if $f_{ij} = 1/2$ we allocate $\ite_j$ to $i$ in exactly one of the two allocations. For brevity, we define $\items_{1/2}$ and $\agents_{1/2}$ as follows:
		\begin{align*}
			\items_{1/2} &= \{\ite_j| \exists i: f_{i,j} = 1/2\}\\
			\agents_{1/2} &= \{i| \exists \ite_j: f_{i,j} = 1/2 \}
		\end{align*}
		
		Consider a bipartite graph $G(X,Y)$ with parts $X$ and $Y$ as follows: For every item $\ite_j \in \items_{1/2}$  there is a vertex $y_j$ in $Y$ corresponding to item $\ite_j$. For every agent $i \in \agents_{1/2}$, we have vertices $x^1_i,x^2_i,\ldots,x_i^{t_i}$ in $X$, where $t_i$ is half the number of items $\ite_j$ such that $f_{ij} = 1/2$, that is
		$$
		t_i = \frac{\big|\{\ite_j| f_{i,j} = 1/2\}\big|}{2}.
		$$
		Also, we add the following edges to $G$. For every $i \in \agents_{1/2}$, order the items $\ite_j$ with $f_{ij} = 1/2$ in decreasing order of their value to agent $i$. Then we connect $x_i^1$ to the first two items, $x^2_i$ to the items with ranks three and four, and so on. In this way, all vertices in $X$ and $Y$ have degree two. 
		Hence, $G$ decomposes into vertex disjoint cycles. Each cycle decomposes into two matchings (note that since the graph is bipartite, all the cycles have even length), and thus $G$ decomposes into two perfect matchings, say $M^1$ and $M^2$. We define allocations $\allocation^1$ and $\allocation^2$ as follows: for every item $\ite_j$, agent $i$, and $r \in \{1,2\}$, we allocate item $\ite_j$ to agent $i$ in $\allocation^r$, if and only if either $f_{i,j} = 1$, or $M^r(y_j) = x_i^k$ for some $1 \leq k \leq t_i$, where $M^r(y_j)$ refers to the vertex matched with $y_j$ in $M^r$. Define $\ralloc$ as a randomized allocation that selects $\allocation^1$ or $\allocation^2$, each with probability $1/2$. 
		
		It is easy to check that for every agent $i$, $\valu_{i}(\ralloc) = \valu_{i}(\falloc)$. Here, we focus on the ex-post guarantee of $\ralloc$. Fix an agent $i$ and Let $b_1$, $b_2$, \ldots, $b_{2t_i}$ be the items half-owned by $i$ in order of decreasing value for agent $i$. Then, by the way we construct $M^1$ and $M^2 $, for all $1 \le \ell \le t_i$, $b_{2\ell - 1}$ and $b_{2\ell}$ are allocated to $i$ in different allocations. Hence, the value of the $i^{\textrm{th}}$ bundle in either allocation $\allocation^r$ satisfies
		\[    v_i(\allocation_i^r) \ge v_i(b_2) + v_i(b_4) + \ldots + v_i(b_{2t_i}) + \sum_{j: f_{i,j}=1} \valu_i(\ite_j) .\]
		We can now bound $v_i(\falloc) - v_i(\allocation_i^r)$ from above. 
		\begin{align*}
			v_i(\falloc) - v_i(\allocation_i^r) &\le \frac{1}{2} \left(\sum_{1 \le \ell \le 2t_i} v_i(b_\ell) \right)  -  \sum_{1 \le \ell \le t_i} v_i(b_{2\ell}) \\
			&= \frac{v_i(b_1)}{2} + \sum_{1 \le \ell < t_i} \left(\frac{v_i(b_{2\ell}) + v_i(b_{2\ell+1})}{2} - v_i(b_{2\ell})\right) + \frac{v_i(b_{2t_i})}{2} - v_i(b_{2t_i})\\
			&\le  \frac{v_i(b_1)}{2}.
		\end{align*}
	\end{proof}
 
	Now we are ready to prove Theorem \ref{thm:2}.
    \ExAnteExPost*
    \begin{proof}
        The ex-ante guarantee follows from Observation \ref{obs:stupper} and Lemma \ref{frac-to-rand}.
        
        Let $\allocation^1$ and $\allocation^2$ be the integral allocations obtained by Lemma \ref{lem:fractorand}. Consider any agent $i$, and let $u_{i,i'}$ be such that $v_i(\ralloc) = \sum_j f_{ij} u_{i,i'}(b_j)$. Then, by Lemma \ref{lem:fractorand} for $r \in \{1,2\}$ we have
	\[  u_{i,i'}(\allocation_i^r) \ge u_{i,i'}(\ralloc_i^\ell) - \frac{\max \{u_{i,i'}(b_j) \, \vert\, f_{i,j} = 1/2\}}{2}\]
	and since by Lemma \ref{ass:1} we know the value of each item for each agent is less than $1/4$, we have 
	\[ v_i(\allocation_i^\ell) \ge v_i(\ralloc_i) - \frac{\max_j v_i(b_j)}{2} > \frac{1}{4} - \frac{1}{8} = \frac{1}{8}.\]
	Hence, the ex-post guarantee holds as well.    
    \end{proof}

\section{\boldmath $3/13$-$\MMS$ Allocation}\label{sec:expost}
In this section, we improve the best approximation guarantee of $\MMS$ for deterministic allocations in the fractionally subadditive setting. We show that a factor $3/13 \approx 0.230769$ of the maximin-share of every agent is possible. Before this work, the best approximation guarantee for maximin-share in the $\XOS$ setting was $0.219225$-$\MMS$~\cite{seddighin2022improved}.

Our algorithm for improving the ex-post guarantee is based on our previous algorithms plus two additional steps and a more  in-depth analysis. In this algorithm, before finding the allocation that maximizes social welfare, we strengthen our upper bound on the value of items. For this, we add two more steps to our algorithm in which we satisfy some of the agents with two items and three items. In contrast to the first step (i.e., allocating single items to agents), these steps might decrease the maximin-share value of the remaining agents for the remaining items. 
Let $t =6/13$. The goal is to find a $t/2$-$\MMS$ allocation. Our allocation algorithm is as follows:
\begin{enumerate}
	\item \label{step-1} While there exists an item $\ite_{j}$ with value at least $t/2$ to an agent $i$, allocate $\ite_j$ to agent $i$ and remove $i$ and $\ite_j$ respectively from $\agents$ and $\items$. 
	\item \label{step-2} While there exists a pair of items $\ite_{j},\ite_k$ with total value of at least $t/2$ to some agent $i$, allocate $\{\ite_{j},\ite_k\}$ to agent $i$, remove both goods from $\items$, and remove agent $i$ from $\agents$. 
    \item \label{step-3} While there exists a triple of items $\ite_j,\ite_k, \ite_s$ with total value of at least $t/2$ to some agent $i$, allocate $\{\ite_{j},\ite_k, \ite_s\}$ to agent $i$, remove all three goods from $\items$, and remove agent $i$ from $\agents$. 
    \item \label{step-4} For the remaining agents $\agents'$ and items $\items'$, proceed as follows: for every agent $i$, define $\bar{\valu}_i$ as follows: 
	$$\forall S \subseteq \items \quad \bar{\valu}_i(S) = \min(t,\valu_{i}(S)).$$
	Let $\bar{v} = (\bar{v}_1, \ldots, \bar{v}_n)$ and return an allocation $\allocation$ that maximizes social welfare with respect to $\bar{v}$, i.e., $\allocation = \arg \max_{A \in \Pi} \sum_{i \in \agents'} \bar{v}_i(A_i)$ where $\Pi$ is the set of all allocations of $\items'$ to $\agents'$. 
\end{enumerate}
\begin{algorithm} [tb]
	\caption{$\mathtt{approxMMS}(\agents, \items, \mathcal{V})$ 
        \\ \textbf{Input:} Instance $(\agents, \items, \mathcal{V})$.
        \\ \textbf{Output:} Allocation $\allocation$. 
    }
    \label{algo:det}
	\begin{algorithmic}[1]
        \State Let $t = 6/13$
        \While{there exists $\ite_j \in \items$ and $i \in \agents$ s.t. $v_i(\ite_j) \geq t/2$} \Comment{Step 1}
            \State $\allocation_i \leftarrow \{\ite_j\}$
            \State $\items \leftarrow \items \setminus \{\ite_j\}$
            \State $\agents \leftarrow \agents \setminus \{i\}$
        \EndWhile

        \While{there exists $\ite_j, \ite_k \in \items$ and $i \in \agents$ s.t. $v_i(\{\ite_j,\ite_k\}) \geq t/2$} \Comment{Step 2}
            \State $\allocation_i \leftarrow \{\ite_j, \ite_k\}$
            \State $\items \leftarrow \items \setminus \{\ite_j, \ite_k\}$
            \State $\agents \leftarrow \agents \setminus \{i\}$
        \EndWhile

        \While{there exists $\ite_j, \ite_k, \ite_s \in \items$ and $i \in \agents$ s.t. $v_i(\{\ite_j, \ite_k, \ite_s\}) \geq t/2$} \Comment{Step 3}
            \State $\allocation_i \leftarrow \{\ite_j, \ite_k, \ite_s\}$
            \State $\items \leftarrow \items \setminus \{\ite_j, \ite_k, \ite_s\}$
            \State $\agents \leftarrow \agents \setminus \{i\}$
        \EndWhile
        \State Let $\agents' = \agents$, $\items' = \items$ and $\bar{\valu}_i(\cdot) = \min(t,\valu_{i}(\cdot))$
        \State Let $\Pi$ be the set of all allocations of $\items'$ to $\agents'$
        \State Let $\allocation = \arg \max_{A \in \Pi} \sum_{i \in \agents'} \bar{v}_i(A_i)$ \Comment{Step 4}
        \State Return $\allocation$
	\end{algorithmic}	
\end{algorithm}

In the rest of this section, we analyze the above algorithm. See Algorithm \ref{algo:det} for the pseudocode. By Lemma \ref{ass:1}, after Step \ref{step-1}, the $\MMS$ value of all the agents is at least $1$. Let $n$ be the number of remaining agents after Step \ref{step-1}.
We denote by $n_1$ and $n_2$, the number of agents that are satisfied in Steps \ref{step-2} and \ref{step-3} respectively and let $n' = n-n_1-n_2 = |\agents'|$ be the number of remaining agents after Step \ref{step-3}. In contrast to the first step, Step \ref{step-2} and \ref{step-3} might decrease the maximin-share value of the remaining agents for the remaining items. However, we prove that the remaining items satisfy certain special structural properties. 

\begin{observation}\label{obs:4:1}
	Since no item can satisfy any remaining agent after Step \ref{step-1}, for every agent $i$ and every item $\ite_j$, we have 
	$\valu_{i}(\{\ite_j \})<t/2	$. 
\end{observation}

Also, by the method that we allocate the items in Step \ref{step-3},
after this step the following observation holds.
\begin{observation}\label{obs:5:1}
	Since after Step \ref{step-3}, no triple of items can satisfy an agent, for every different items $\ite_j,\ite_k,\ite_s$ and every agent $i$ we have $\valu_{i}(\{\ite_j,\ite_k,\ite_s\}) < t/2$.
\end{observation}
Note that since the valuations are $\XOS$, Observation \ref{obs:5:1} implies no upper bound better than $t/2$ on the value of a single item to an agent. For example, consider the following extreme scenario: for a small constant $\epsilon>0$, the value of every non-empty subset of items to agent $i$ is equal to $t/2-\epsilon$. It is easy to check that this valuation function is $\XOS$. For this case, the value of every triple of items is also equal to $t/2-\epsilon$, but this implies no upper bound  better than $t/2$ on the value of a single item.

\begin{lemma}\label{cl:1:1}
	Fix a remaining agent $i$ and consider the $n$ bundles with value at least $1$ in an $\MMS$ partition of agent $i$ after Step \ref{step-1}. 
	Put these bundles into  4 different sets $B_0, B_1, B_2, B_{\geq 3}$, where for $0 \leq \ell \leq 2$, set $B_\ell$ contains bundles that lose exactly $\ell$ items in Steps \ref{step-2} and \ref{step-3}, and $B_{\geq 3}$ contains bundles that lose at least three  items in these steps. 
After Step \ref{step-3}, the following inequality holds:
	$$
n' \leq |B_0| + \frac{2}{3}|B_1| + \frac{1}{3}|B_2|.
	$$
\end{lemma}
\begin{proof}
	Since each satisfied agent in step $k$ receives $k$ items, we have:
	$$2n_1 + 3n_2 \geq |B_1| + 2|B_2|+3|B_{\geq 3}|.$$
    Thus, 
    \begin{equation} \label{ineq:1:1}
        n_1 + n_2 \geq \frac{2}{3}n_1 + n_2 
        \geq \frac{1}{3}|B_1| + \frac{2}{3}|B_2| + |B_{\geq 3}|,
      \end{equation}
      and therefore,
    \begin{align*}
        n' &= n-n_1-n_2 \\
        &\leq n - \frac{1}{3}|B_1| - \frac{2}{3}|B_2| - |B_{\geq 3}| \tag{Inequality \ref{ineq:1:1}}\\
        &= |B_0|+\frac{2}{3}|B_1|+\frac{1}{3}|B_2|. \tag{$n=|B_0|+|B_1|+|B_2|+|B_{\geq 3}|$}
    \end{align*}
\end{proof}

Finally, in Step \ref{step-4}, we find the integral allocation $\allocation$ that maximizes social welfare with respect to $\bar{\valu}$ for the remaining agents. Let  
$$
Z = \sum_{i \in \agents'} \bar{\valu}_i(\allocation_i).
$$

Since for each remaining agent $i$, $\bar{v}_i(\allocation_i)$ is upper-bounded by $t$, we have $Z \leq n't$. If for every agent $i$, $\valu_i(\allocation_i)\geq t/2$ holds, then $\allocation$ is $t/2$-$\MMS$, and we are done. Therefore, for the rest of this section, assume that for an agent $i^*$, we have $\valu_{i^*}(\allocation_{i^*})<t/2$.

\begin{lemma}\label{lem:important}
    For all sets $S \subseteq M$, $C^\allocation_{\bar{v}}(S) \geq \bar{v}_{i^*}(S) - \bar{v}_{i^*}(\allocation_{i^*})$.
\end{lemma}
\begin{proof}
    Let allocation $\allocation'$ be as following. For all agents $i$, $\allocation'_i = \allocation_i \setminus S$. Basically, $\allocation'$ is allocation $\allocation$ after removing all the items in $S$ from the bundles they belong to. We have
    \begin{align}\label{ineq:12}
        \sum_{i \in \agents} \bar{v}_{i} (\allocation'_i) &= \sum_{i \in \agents} \bar{v}_{i} (\allocation_i) - C^\allocation_{\bar{v}}(S).
    \end{align}
    Now let $\allocation''$ be allocation $\allocation'$ after allocating $S$ to agent $i^*$. I.e., for all agents $i \neq i^*$, $\allocation''_i = \allocation'_i$ and $\allocation''_{i^*} = \allocation'_{i^*} \cup S$. We have
    \begin{align*}
        \sum_{i \in \agents} \bar{v}_{i} (\allocation_i) &\geq \sum_{i \in \agents} \bar{v}_{i} (\allocation''_i) \tag{$\allocation = \argmax_{A \in \Pi} \sum_{i \in \agents'} \bar{v}_i(A_i)$}\\
        &= \sum_{i \in \agents \setminus \{i^*\}} \bar{v}_{i} (\allocation'_i) + \bar{v}_{i^*}(\allocation'_{i^*} \cup S) \\
        &= \left(\sum_{i \in \agents} \bar{v}_{i} (\allocation_i) - C^\allocation_{\bar{v}}(S) - \bar{v}_{i^*}(\allocation'_{i^*})\right) + \bar{v}_{i^*}(\allocation'_{i^*} \cup S) \tag{Inequality \eqref{ineq:12}}\\
        &\geq \sum_{i \in \agents} \bar{v}_{i} (\allocation_i) - C^\allocation_{\bar{v}}(S) - \bar{v}_{i^*}(\allocation'_{i^*}) + \bar{v}_{i^*}(S). \tag{$\bar{v}_{i^*}(\allocation'_{i^*} \cup S) \geq \bar{v}_{i^*}(S)$}
    \end{align*}
    Therefore, $C^\allocation_{\bar{v}}(S) \geq \bar{v}_{i^*}(S) - \bar{v}_{i^*}(\allocation_{i^*})$.
\end{proof}

Let $B_0,B_1,$ and $B_2$ be the sets defined for agent $i^*$ in Lemma \ref{cl:1:1}. In Lemmas \ref{B0-bound}, \ref{B1-bound} and \ref{B2-bound}, we give lower bounds on the contribution of the bundles in $B_0$, $B_1$ and $B_2$ to $\allocation$ respectively.
\begin{lemma}\label{B0-bound}
    After Step \ref{step-3}, for all bundles $X \in B_0$, there exists a partition of $X$ into $X^1$ and $X^2$ such that $C^{\allocation}_{\bar{v}}(X^1)+C^{\allocation}_{\bar{v}}(X^2) \geq t$.
\end{lemma}
\begin{proof}
    The idea is to partition the set $X$ into two bundles $X^1$ and $X^2$ each with value at least $t$ to agent $i^*$. Then using Lemma \ref{lem:important}, we prove the contribution of each of these bundles to $\allocation$ is at least $t/2$ and thus the total contribution is at least $t$.
    
    For a fixed bundle $X \in B_0$, let $j$ be such that $u_{i^*,j}(X) = v_{i^*}(X) \geq 1$. Let $g_1$ and $g_2$ be two different most valuable items in $X$ with respect to $u_{i^*, j}$, i.e., for all items $g \in X \setminus \{g_1, g_2\}$, $u_{i^*,j}(g_1) \geq u_{i^*,j}(g_2) \geq u_{i^*,j}(g)$. Let $X^1$ be a minimal subset of $X$ such that $\{g_1, g_2\} \subset X^1$ and $u_{i^*, j}(X^1) \geq t$. Let $X^2 = X \setminus X^1$. Since $X^1$ is minimal, for all $g \in X_1$, $u_{i^*, j}(X^1 \setminus \{g\}) < t$. Also, by Observation \ref{obs:5:1}, for all $g \in X^1 \setminus \{g_1,g_2\}$, $u_{i^*,j}(\{g_1, g_2, g\}) \leq v_{i^*}(\{g_1, g_2, g\}) < t/2$ and thus, $u_{i^*, j}(g) < t/6$. Therefore, for all $g \in X^1 \setminus \{g_1,g_2\}$,
    \begin{align*}
        u_{i^*, j}(X^2) &\geq 1 - u_{i^*, j}(X^1) \tag{$u_{i^*, j}(X^1 \cup X^2) \geq 1$}\\
        &= 1- \left(u_{i^*, j}(X^1 \setminus \{g\}) + u_{i^*, j}(g)\right) \tag{by additivity of $u_{i^*,j}$}\\
        &> 1- \frac{7}{6}t \\
        &= t. \tag{$t=6/13$}
    \end{align*}
    Hence, we have $v_{i^*}(X^1) \geq u_{i^*, j}(X_1) \geq t$ and $v_{i^*}(X^2) \geq u_{i^*, j}(X_2) \geq t$. Now by Lemma \ref{lem:important}, we have
    \begin{align*}
        C^{\allocation}_{\bar{v}}(X^1)+C^{\allocation}_{\bar{v}}(X^2) &\geq \left( \bar{v}_{i^*}(X^1) - \bar{v}_{i^*}(\allocation_{i^*}) \right) + \left( \bar{v}_{i^*}(X^2) - \bar{v}_{i^*}(\allocation_{i^*}) \right) \\
        &> 2(t-\frac{1}{2}t) = t.
    \end{align*}
\end{proof}

\begin{lemma}\label{B1-bound}
    After Step \ref{step-3}, for all bundles $X \in B_1$, there exists a partition of $X$ into $X^1$ and $X^2$ such that $C^{\allocation}_{\bar{v}}(X^1)+C^{\allocation}_{\bar{v}}(X^2) \geq \frac{2}{3}t$.
\end{lemma}
\begin{proof}
    Fix a bundle $X \in B_1$. By Lemma \ref{ass:1}, the $\MMS$ value of agent $i^*$ is at least $1$ after Step \ref{step-1}. By Observation \ref{obs:4:1}, $v_{i^*}(g) < t/2$ for all remaining items $g$ after Step \ref{step-1}. Since $X$ is a bundle in an $\MMS$ partition of agent $i^*$ after Step \ref{step-1} and after the removal of one item $g$, we have 
    \begin{align}
        v_{i^*}(X) > 1-t/2.
    \end{align}
    Let $j$ be such that $u_{i^*,j}(X) = v_{i^*}(X)$.
    Let $g_1$ and $g_2$ be two different most valuable items in $X$ with respect to $u_{i^*, j}$, i.e., for all items $g \in X \setminus \{g_1, g_2\}$, $u_{i^*,j}(g_1) \geq u_{i^*,j}(g_2) \geq u_{i^*,j}(g)$. Let $X^1$ be a minimal subset of $X$ such that $\{g_1, g_2\} \subset X^1$ and $u_{i^*, j}(X^1) \geq 2t/3$. Let $X^2 = X \setminus X^1$. Since $X^1$ is minimal, for all $g \in X_1$, $u_{i^*, j}(X^1 \setminus \{g\}) < 2t/3$. Also, by Observation \ref{obs:5:1}, for all $g \in X^1 \setminus \{g_1,g_2\}$, $u_{i^*,j}(\{g_1, g_2, g\}) \leq v_{i^*}(\{g_1, g_2, g\}) < t/2$ and thus, $u_{i^*, j}(g) < t/6$. Therefore,
    \begin{align*}
        u_{i^*, j}(X_1) &= u_{i^*, j}(X_1 \setminus \{g\}) + u_{i^*, j}(g) \tag{by additivity of $u_{i^*, j}$}\\
        &< \frac{2}{3}t + \frac{1}{6}t = \frac{5}{6}t.
    \end{align*}
    Therefore, for all $g \in X^1 \setminus \{g_1,g_2\}$, we have
    \begin{align*}
        C^{\allocation}_{\bar{v}}(X^1)+C^{\allocation}_{\bar{v}}(X^2) &\geq \left( \bar{v}_{i^*}(X^1) - \bar{v}_{i^*}(\allocation_{i^*}) \right) + \left( \bar{v}_{i^*}(X^2) - \bar{v}_{i^*}(\allocation_{i^*}) \right) \tag{Lemma \ref{lem:important}}\\
        &= \left(\min(t, v_{i^*}(X^1)) - \bar{v}_{i^*}(\allocation_{i^*})\right) + \left(\min(t, v_{i^*}(X^2)) - \bar{v}_{i^*}(\allocation_{i^*})\right) \\
        &> \left(\min(t, u_{i^*, j}(X^1)) - \frac{1}{2}t\right) + \left(\min(t, u_{i^*, j}(X^2)) - \frac{1}{2}t \right) \\
        &\geq u_{i^*, j}(X^1) + \min(t, 1-\frac{1}{2}t - u_{i^*, j}(X^1)) - t \tag{$u_{i^*, j}(X^1) < 5t/6$}\\
        &\geq \min (u_{i^*, j}(X^1), 1-\frac{3}{2}t) \\
        &\geq \frac{2}{3}t.
    \end{align*}
\end{proof}

\begin{lemma}\label{B2-bound}
    After Step \ref{step-3}, for all bundles $X \in B_2$, $C^\allocation_{\bar{v}}(X)\geq \frac{1}{2}t$.
\end{lemma}
\begin{proof}
    Fix a bundle $X \in B_2$. By Lemma \ref{ass:1}, the $\MMS$ value of agent $i^*$ is at least $1$ after Step \ref{step-1}. By Observation \ref{obs:4:1}, $v_{i^*}(\{g\}) < t/2$ for all remaining items $g$ after Step \ref{step-1}. Since $X$ is a bundle in an $\MMS$ partition of agent $i^*$ after Step \ref{step-1} and after the removal of two items like $g$, we have $v_{i^*}(X) > 1-t > t$.
    Therefore, $\bar{v}_{i^*}(X)=\min(t,v_{i^*}(X))=t$.
    Now by Lemma \ref{lem:important},
    \begin{align*}
        C^\allocation_{\bar{v}}(X) \geq \bar{v}_{i^*}(X) - \bar{v}_{i^*}(\allocation_i) > t-\frac{1}{2}t = \frac{1}{2}t.
    \end{align*}
\end{proof}

\ExPost*
\begin{proof}
    Let $\allocation$ be the output of Algorithm \ref{algo:det}. Towards a contradiction, assume for agent $i^*$, $v_{i^*}(\allocation_{i^*}) < 3/13 = t/2$.
    For all agents $i$ which are removed during the first three steps, we have $v_i(\allocation_i) \geq t/2=3/13$. Therefore, $i^* \in \agents'$. For all $X \in B_0$, let $X^1$ and $X^2$ be as defined in Lemmas \ref{B0-bound} and \ref{B1-bound}. We have
    \begin{align*}
        t(n'-\frac{1}{2}) &> \sum_{i \in \agents'} \bar{v}_i(\allocation_i) \tag{for all $i \in \agents', \bar{v}_i(\allocation_i) \leq t$ and $\bar{v}_{i^*}(\allocation_{i^*}) < t/2$}\\
        &\geq \sum_{X \in B_0} \left(C^\allocation_{\bar{v}}(X^1) + C^\allocation_{\bar{v}}(X^2) \right) + \sum_{X \in B_1} \left(C^\allocation_{\bar{v}}(X^1) + C^\allocation_{\bar{v}}(X^2) \right) + \sum_{X \in B_2} C^\allocation_{\bar{v}}(X) \tag{Lemma \ref{lem:supper-additive}}\\
        &\geq t|B_0| + \frac{2}{3}t|B_1| + \frac{1}{2}t|B_2| \tag{Lemmas \ref{B0-bound}, \ref{B1-bound} and \ref{B2-bound}}\\
        &\geq tn', \tag{Lemma \ref{cl:1:1}}
    \end{align*}
    which is a contradiction. Therefore, such an agent $i^*$ does not exist and $\allocation$ is a $3/13$-$\MMS$~allocation.
\end{proof}

\section{Future Directions}
We developed randomized and deterministic allocations guaranteeing approximations of maximin-share for fractionally subadditive valuations. For randomized allocations, we derived simulataneous ex-ante and ex-post guarantees. Several interesting questions remain open. 

The most straight-forward direction is to improve  approximation guarantees for both ex-ante and ex-post cases. The first result we obtained in this paper is an allocation that is simultaneously $1/4$-$\MMS$ ex-ante and $1/8$-$\MMS$ ex-post. Also, we proved the existence of an allocation which is $3/13$-$\MMS$ ex-post. None of these results are known to be tight. Therefore, the following questions remain open.
   
\begin{question} 	
Can we find a randomized allocation with ex-ante $\MMS$ approximation guarantee better than 1/4 for fractionally subadditive valuations?
\end{question}

Note that the result of Lemma \ref{lem:1} shows that no randomized allocation can guarantee a fraction better than $3/4$-$\MMS$ ex-ante. Therefore, a gap of $[1/4,3/4)$ remains between the best upper bound and the best lower bound for guaranteeing $\MMS$ ex-ante for fractionally subadditive valuations.

\begin{question} 	
Can we find an allocation with ex-post $\MMS$ approximation guarantee better than $3/13$?
\end{question}

Moreover, for simultaneous ex-ante and ex-post guarantees there might be room for improving Theorem \ref{thm:2}. 
\begin{question} 	
Can we find better allocations guaranteeing $\MMS$ both ex-ante and ex-post simultaneously?
\end{question}

We believe that extending the idea in the second step of the proof of Theorem \ref{thm:1} might help improving the ex-post guarantee: we can add additional steps to check if some agent can be satisfied with two, three or more items. This way, we would have better upper bounds on the value of remaining items that head to the next step. However, having such additional steps need a more careful analysis of the maximin-share value of the remaining agents. 

Another notable point about the results of this paper is that the algorithms defined in Theorems \ref{thm:1}, \ref{thm:2} and \ref{thm:ExPost} are not necessarily implementable in polynomial time. 
The reason is that calculating the exact value of $\MMS_i$ for each agent and also calculating an $\MMS$ partition of the agents is not implementable in polynomial time. Currently, the best polynomial time ex-post guarantee for $\MMS$ in the $\XOS$ setting is $1/8$ \cite{ghodsi2018fair}.
This result can also be considered as the best polynomial time algorithm for finding an allocation with ex-ante guarantee for $\MMS$ in the additive case. 

\begin{question} 
Can we find an allocation with a constant ex-ante $\MMS$ approximation guarantee better than $1/8$ that can be implemented in polynomial time? 
\end{question}

Also, an interesting open direction is to find similar results for other classes of valuation functions, including submodular and subadditive valuations. 

\begin{question} 
Can we find simultaneous ex-ante and ex-post guarantees for  $\MMS$ for submodular and subadditive set functions?
\end{question}

Currently the best guarantee for the submodular setting is $10/27$-$\MMS$ ex-post \cite{uziahu2023fair}. Also, for the subadditive case, the best known ex-post approximation guarantee is $O(\frac{1}{\log n \log \log n})$.

\newpage
\bibliographystyle{apalike} 
\bibliography{efr}

\begin{thebibliography}{}

\bibitem[Akrami and Garg, 2024]{Akrami2023BreakingT3}
Akrami, H. and Garg, J. (2024).
\newblock Breaking the 3/4 barrier for approximate maximin share.
\newblock In {\em Proceedings of the 2024 {ACM-SIAM} Symposium on Discrete
  Algorithms, {SODA}, 2024}.

\bibitem[Akrami et~al., 2023a]{akrami2023simplification}
Akrami, H., Garg, J., Sharma, E., and Taki, S. (2023a).
\newblock Simplification and improvement of {MMS} approximation.
\newblock {\em CoRR}, abs/2303.16788.

\bibitem[Akrami et~al., 2023b]{Akrami2023ImprovingAG}
Akrami, H., Garg, J., and Taki, S. (2023b).
\newblock Improving approximation guarantees for maximin share.
\newblock {\em CoRR}, abs/2307.12916.

\bibitem[Aleksandrov et~al., 2015]{DBLP:conf/ijcai/AleksandrovAGW15}
Aleksandrov, M., Aziz, H., Gaspers, S., and Walsh, T. (2015).
\newblock Online fair division: Analysing a food bank problem.
\newblock In {\em {IJCAI} 2015, Buenos Aires, Argentina, July 25-31, 2015},
  pages 2540--2546. {AAAI} Press.

\bibitem[Amanatidis et~al., 2019]{amanatidis2019multiple}
Amanatidis, G., Ntokos, A., and Markakis, E. (2019).
\newblock Multiple birds with one stone: Beating $1/2$ for efx and gmms via
  envy cycle elimination.
\newblock {\em arXiv preprint arXiv:1909.07650}.

\bibitem[Aziz, 2020]{DBLP:journals/corr/abs-2002-10171}
Aziz, H. (2020).
\newblock A probabilistic approach to voting, allocation, matching, and
  coalition formation.
\newblock {\em CoRR}, abs/2002.10171.

\bibitem[Aziz et~al., 2023]{aamas23Aziz}
Aziz, H., Ganguly, A., and Micha, E. (2023).
\newblock Best of both worlds fairness under entitlements.
\newblock {\em AAMAS}.

\bibitem[Babaioff et~al., 2021]{DBLP:conf/aaai/BabaioffEF21}
Babaioff, M., Ezra, T., and Feige, U. (2021).
\newblock Fair and truthful mechanisms for dichotomous valuations.
\newblock In {\em {AAAI} 2021, {IAAI} 2021, {EAAI} 2021, Virtual Event,
  February 2-9, 2021}, pages 5119--5126. {AAAI} Press.

\bibitem[Babaioff et~al., 2022]{babaioff2022best}
Babaioff, M., Ezra, T., and Feige, U. (2022).
\newblock On best-of-both-worlds fair-share allocations.
\newblock In {\em WINE 2022, Troy, NY, USA, December 12--15, 2022,
  Proceedings}, pages 237--255. Springer.

\bibitem[Barman and Krishna~Murthy, 2017]{barman2017approximation}
Barman, S. and Krishna~Murthy, S.~K. (2017).
\newblock Approximation algorithms for maximin fair division.
\newblock In {\em Proceedings of the 2017 ACM Conference on Economics and
  Computation}, pages 647--664. ACM.

\bibitem[Barman and Krishnamurthy, 2020]{DBLP:journals/teco/BarmanK20}
Barman, S. and Krishnamurthy, S.~K. (2020).
\newblock Approximation algorithms for maximin fair division.
\newblock {\em {ACM} Trans. Economics and Comput.}, 8(1):5:1--5:28.

\bibitem[Bogomolnaia and Moulin, 2001]{DBLP:journals/jet/BogomolnaiaM01}
Bogomolnaia, A. and Moulin, H. (2001).
\newblock A new solution to the random assignment problem.
\newblock {\em J. Econ. Theory}, 100(2):295--328.

\bibitem[Budish, 2011]{Budish:first}
Budish, E. (2011).
\newblock The combinatorial assignment problem: Approximate competitive
  equilibrium from equal incomes.
\newblock {\em Journal of Political Economy}, 119(6):1061--1103.

\bibitem[Caragiannis et~al., 2019]{caragiannis2019unreasonable}
Caragiannis, I., Kurokawa, D., Moulin, H., Procaccia, A.~D., Shah, N., and
  Wang, J. (2019).
\newblock The unreasonable fairness of maximum nash welfare.
\newblock {\em ACM Transactions on Economics and Computation (TEAC)}, 7(3):12.

\bibitem[Dehghani et~al., 2018]{dehghani2018envy}
Dehghani, S., Farhadi, A., HajiAghayi, M., and Yami, H. (2018).
\newblock Envy-free chore division for an arbitrary number of agents.
\newblock In {\em Proceedings of the Twenty-Ninth Annual ACM-SIAM Symposium on
  Discrete Algorithms}, pages 2564--2583. SIAM.

\bibitem[Dickerson et~al., 2014]{dickerson2014computational}
Dickerson, J.~P., Goldman, J., Karp, J., Procaccia, A.~D., and Sandholm, T.
  (2014).
\newblock The computational rise and fall of fairness.
\newblock In {\em Twenty-Eighth AAAI Conference on Artificial Intelligence}.

\bibitem[Farhadi et~al., 2019]{asymmetric}
Farhadi, A., Ghodsi, M., Hajiaghayi, M.~T., Lahaie, S., Pennock, D.~M.,
  Seddighin, M., Seddighin, S., and Yami, H. (2019).
\newblock Fair allocation of indivisible goods to asymmetric agents.
\newblock {\em J. Artif. Intell. Res.}, 64:1--20.

\bibitem[Feige, 2009]{DBLP:journals/siamcomp/Feige09}
Feige, U. (2009).
\newblock On maximizing welfare when utility functions are subadditive.
\newblock {\em {SIAM} J. Comput.}, 39(1):122--142.

\bibitem[Feige et~al., 2022]{feige2022tight}
Feige, U., Sapir, A., and Tauber, L. (2022).
\newblock A tight negative example for mms fair allocations.
\newblock In {\em Web and Internet Economics: 17th International Conference,
  WINE 2021, Potsdam, Germany, December 14--17, 2021, Proceedings}, pages
  355--372. Springer.

\bibitem[Feldman et~al., 2023]{cycle-breaking}
Feldman, M., Mauras, S., Narayan, V.~V., and Ponitka, T. (2023).
\newblock Breaking the envy cycle: Best-of-both-worlds guarantees for
  subadditive valuations.
\newblock {\em CoRR}, abs/2304.03706.

\bibitem[Foley, 1967]{Foley:first}
Foley, D.~K. (1967).
\newblock Resource allocation and the public sector.
\newblock {\em YALE ECON ESSAYS, VOL 7, NO 1, PP 45-98, SPRING 1967. 7 FIG, 13
  REF.}

\bibitem[Freeman et~al., 2020]{freeman2020best}
Freeman, R., Shah, N., and Vaish, R. (2020).
\newblock Best of both worlds: Ex-ante and ex-post fairness in resource
  allocation.
\newblock In {\em Proceedings of the 21st ACM Conference on Economics and
  Computation}, pages 21--22.

\bibitem[Garg and Taki, 2021]{DBLP:journals/ai/GargT21}
Garg, J. and Taki, S. (2021).
\newblock An improved approximation algorithm for maximin shares.
\newblock {\em Artif. Intell.}, 300:103547.

\bibitem[Ghodsi et~al., 2018]{ghodsi2018fair}
Ghodsi, M., HajiAghayi, M., Seddighin, M., Seddighin, S., and Yami, H. (2018).
\newblock Fair allocation of indivisible goods: Improvements and
  generalizations.
\newblock In {\em Proceedings of the 2018 ACM Conference on Economics and
  Computation}, pages 539--556. ACM.

\bibitem[Halpern et~al., 2020]{DBLP:conf/wine/0002PP020}
Halpern, D., Procaccia, A.~D., Psomas, A., and Shah, N. (2020).
\newblock Fair division with binary valuations: One rule to rule them all.
\newblock In {\em {WINE} 2020, Beijing, China, December 7-11, 2020,
  Proceedings}.

\bibitem[Hoefer et~al., 2023]{DBLP:journals/corr/abs-2209-03908}
Hoefer, M., Schmalhofer, M., and Varricchio, G. (2023).
\newblock Best of both worlds: Agents with entitlements.
\newblock In {\em Proceedings of the 2023 International Conference on
  Autonomous Agents and Multiagent Systems, {AAMAS} 2023, London, United
  Kingdom, 29 May 2023 - 2 June 2023}, pages 564--572. {ACM}.

\bibitem[Kurokawa et~al., 2018]{kurokawa2018fair}
Kurokawa, D., Procaccia, A.~D., and Wang, J. (2018).
\newblock Fair enough: Guaranteeing approximate maximin shares.
\newblock {\em Journal of the ACM (JACM)}, 65(2):8.

\bibitem[Li and Vetta, 2018]{DBLP:conf/wine/LiV18}
Li, Z. and Vetta, A. (2018).
\newblock The fair division of hereditary set systems.
\newblock In {\em {WINE} 2018, Oxford, UK, December 15-17, 2018, Proceedings}.

\bibitem[Nash~Jr, 1950]{nash1950bargaining}
Nash~Jr, J.~F. (1950).
\newblock The bargaining problem.
\newblock {\em Econometrica: Journal of the Econometric Society}, pages
  155--162.

\bibitem[Plaut and Roughgarden, 2020]{DBLP:journals/siamdm/PlautR20}
Plaut, B. and Roughgarden, T. (2020).
\newblock Almost envy-freeness with general valuations.
\newblock {\em {SIAM} J. Discret. Math.}, 34(2):1039--1068.

\bibitem[Procaccia and Wang, 2014]{10.1145/2600057.2602835}
Procaccia, A.~D. and Wang, J. (2014).
\newblock Fair enough: Guaranteeing approximate maximin shares.
\newblock EC '14, page 675–692, New York, NY, USA. Association for Computing
  Machinery.

\bibitem[Seddighin and Seddighin, 2022]{seddighin2022improved}
Seddighin, M. and Seddighin, S. (2022).
\newblock Improved maximin guarantees for subadditive and fractionally
  subadditive fair allocation problem.
\newblock In {\em {AAAI} 2022, {IAAI} 2022, {EAAI} 2022 Virtual Event, February
  22 - March 1, 2022}. {AAAI} Press.

\bibitem[Steinhaus, 1948]{Steinhaus:first}
Steinhaus, H. (1948).
\newblock The problem of fair division.
\newblock {\em Econometrica}, 16(1).

\bibitem[Uziahu and Feige, 2023]{uziahu2023fair}
Uziahu, G.~B. and Feige, U. (2023).
\newblock On fair allocation of indivisible goods to submodular agents.

\bibitem[Varian, 1973]{varian1973equity}
Varian, H.~R. (1973).
\newblock Equity, envy, and efficiency.

\end{thebibliography}
\newpage


\end{document}